\documentclass{article}

\IfFileExists{iclr2027_conference.sty}{%
  \usepackage{iclr2027_conference,times}%
}{%
  \usepackage[margin=1in]{geometry}%
  \newcommand{\iclrfinalcopy}{}%
}

\usepackage{amsmath,amssymb,amsthm}
\usepackage{booktabs}
\usepackage{graphicx}
\usepackage{float}
\usepackage{flafter}
\usepackage{xcolor}
\usepackage{multirow}
\usepackage[hidelinks]{hyperref}
\usepackage{microtype}
\usepackage{enumitem}
\usepackage{array}       % >{...} column specifiers in the appendix tables
\usepackage{algorithm}   % Algorithm 1 (the selection protocol)
\usepackage{algorithmic}

\newtheorem{proposition}{Proposition}

\newcommand{\onev}{\mathbf{1}}

\newcommand{\funcset}{\mathcal{T}}

\newcommand{\norm}[1]{\lVert #1 \rVert}
\newcommand{\abs}[1]{\lvert #1 \rvert}

\title{Contextual Utility of Quantization Moves\\in Extreme Low-Bit LLMs}

\ifdefined\anonymous
  \author{Anonymous authors\\Paper under double-blind review}
\else
  \author{%
  Wenxuan Xiao\thanks{Equal contribution.}\\
  Astrmira Tech.\\
  \texttt{w.xiao@astrmira.com}
  \And
  Xu Cao\footnotemark[1]\\
  Astrmira Tech.\\
  \texttt{x.cao@astrmira.com}%
  }
  \iclrfinalcopy
\fi
\begin{document}
\maketitle

% The ICLR style sets the running head inside \maketitle, so any override
% must come after it.  In the named preprint build \iclrfinalcopy would
% otherwise claim acceptance, so we state what the document actually is.
% The anonymous build keeps the style's own "Under review" head.
\ifdefined\anonymous\else
  \lhead{Preprint. Under review.}
\fi

\begin{abstract}
Accurate measurements of isolated weight replacements need not predict
their joint effect on a quantized model. We study this gap for legal
replacements at a fixed bitrate, separating isolated pricing error,
composition error and changes in evaluation data. On Llama-3.2-1B and
activation-ordered 3B GPTQ bases, option-KL composition residuals are
$26\times$ and $42\times$ the sums of absolute midpoint-pricing errors;
large residuals remain with exact-singleton prices. Exhaustive Qwen3-0.6B
lattices reveal context-dependent effects: six of eight moves in the
stratified lattice have statistically resolved ARC-KL sign reversals,
excluding an exact monotone additive representation. Pair measurements
improve prediction of unqueried combinations on both Qwen and Llama.
Their value for selection depends on the criterion: the stratified Qwen
front retains a quality-coverage advantage after refitting, whereas the
3B additive and quadratic fronts are nearly equal in quality at a
one-standard-error tolerance. Low-order endpoint structure leaves
continuous gradient responses undetermined, and shared-midpoint prices
can differ from those of the executed subset. Support-matched
controls show that restricting the number of replacements curbs the
failures of large combinations, while additional endpoint search gives
bank-dependent NLL/ARC trade-offs.
\end{abstract}

\section{Introduction}
\label{sec:intro}

Quantization solvers choose a discrete weight state. GPTQ, AWQ and
OmniQuant~\citep{frantar2023gptq,lin2024awq,shao2024omniquant} make this choice
using reconstruction objectives. Direct loss measurements can value a
replacement more faithfully, but an isolated measurement holds the rest of
the model fixed. Executing several replacements changes that context.
Even exact individual effects may then give the wrong prediction for the
resulting state.

Nonadditivity of isolated layer losses is observed in PQI~\citep{hu2025pqi}.
CLADO measures pair effects that change the preferred bit allocation, and
CoopQ uses interaction-aware allocation~\citep{deng2025clado,zhao2025coopq}.
A monotone additive index can still approximate configuration rankings when
some contextual orderings reverse~\citep{hill2026coverage}.
For state selection, the question is what information beyond isolated losses
changes the quality of the chosen combinations. We study
\emph{legal replacements at a fixed bitrate}. Each replacement retains its
encoding and displacement as we vary the surrounding state, so the
comparison isolates context without changing the bit allocation.

We first separate inaccurate prices from failure of additivity on identical
evaluation units. On Llama-3.2-1B, midpoint prices for 61 projection
replacements have a net option-KL error of $0.005$, whereas their composition
adds a $0.459$ residual. The residual is $26\times$ the \emph{sum of absolute}
isolated pricing errors. On an activation-ordered Llama-3.2-3B GPTQ base,
the ratio is $42\times$. Exact-singleton prices still predict impossible
negative KL values for the selected combinations. These controls locate
the error in composition separately from both numerical pricing and data shift.

A move's effect in a given surrounding state is its \emph{contextual utility}.
On the stratified Qwen3 lattice, six of eight moves reverse their ARC effect
across contexts, with simultaneous intervals excluding zero on both sides.
Pair information predicts much of the resulting attainable set: a quadratic
fitted to subsets of at most three moves recovers $92\%$ of the large-subset
Pareto front, versus $52\%$ for additivity. We evaluate whether this gain
also changes the losses attainable on the predicted front and the utility
of one selected state. The Qwen stratified front retains a quality-coverage
advantage after refitting, while its single-state regret difference remains
unresolved. On a separate 3B GPTQ lattice, pair measurements improve
prediction even though the two fronts have nearly equal quality at one
standard error.

Shared-midpoint gradients probe a continuous path whose surrounding proposal
can include partners absent from the executed subset. We show that accurate
endpoint quadratic models do not identify these gradient responses.
We then compare restricting the
number of replacements with querying their combinations directly.
Restricting addition to a few changes avoids much of its large-support
failure. On new 3B NLL data, a two-change beam improves over the base,
while its advantage over two-change addition remains unresolved.

\section{Isolating the Error of Composition}
\label{sec:setup}

\paragraph{Legal states and moves.}
A low-bit \emph{hard state} specifies the codes and metadata needed to decode
every weight group. The Qwen banks use a sign bit, a strong/weak bit and two
FP32 magnitudes per group of 128 (2.5 bits per quantized weight including
metadata). The Llama banks use asymmetric INT3 with a scale and zero point
per group of 128. A replacement takes the complete encoding tuple from one
candidate; mixing disjoint groups preserves legality and bitrate.
Appendix~\ref{app:setup} specifies the formats, models and data.

A proposal $d=q_1-q_0$ is partitioned into $n$ disjoint moves,
$d=\sum_i d_i$, with $D=[d_1\cdots d_n]$. Every $x\in\{0,1\}^n$ defines a
legal state $q_0+Dx$, giving a Boolean lattice of $2^n$ states. In categorical banks,
the alternatives $A_p$ for projection $p$ satisfy $\sum_{i\in A_p}x_i\le1$;
choosing none retains the base encoding. A set's \emph{support} is its number
of executed moves. Write $\mathcal X$ for the legal
choices: the full Boolean cube for disjoint moves, and the constrained cube
for categorical alternatives. Each experiment specifies its bank construction
and sublattice selection.
For a declared loss $F_j$, define
\begin{equation}
 f_j(x)=F_j(q_0+Dx)-F_j(q_0),\qquad x\in\mathcal X .
 \label{eq:setfunction}
\end{equation}
We use next-token negative log-likelihood (NLL) and teacher-KL, including KL
over task options. Negative differences mean improvement. A Pareto state has
no competitor with every loss no larger and at least one smaller.
The marginal effect of adding $i$ after $S$ is
$f_j(S\cup\{i\})-f_j(S)$; its isolated effect is $f_j(\{i\})$.

\paragraph{Two errors on identical data.}
Let $\hat s_i$ be the price of move $i$, $s_i$ its exact isolated FIT loss
difference, and $\Delta^{\mathrm{fit}}(S)$ the exact difference after executing $S$.
Here exact refers to measured losses on the fixed evaluation units.
Then
\begin{equation}
 \Delta^{\mathrm{fit}}(S)-\sum_{i\in S}\hat s_i
 =\underbrace{\sum_{i\in S}(s_i-\hat s_i)}_{\text{pricing error}}
 +\underbrace{\Delta^{\mathrm{fit}}(S)-\sum_{i\in S}s_i}_{\text{composition error}}.
 \label{eq:decomp}
\end{equation}
The first term is the error in individual prices. The second is the error
from adding the true isolated effects. Predicting TEST adds the
separate term $\Delta^{\mathrm{test}}(S)-\Delta^{\mathrm{fit}}(S)$.
For a fixed $S$, exact-singleton pricing removes the first term and leaves
the second unchanged.

\begin{table}[htb]
\centering
\caption{Composition error after accurate or exact isolated pricing.
Within each bank, all columns use identical FIT ARC-Easy option-KL units. Additive selection
minimizes summed ARC prices subject to summed NLL non-increase.
The 1B GPTQ base has no activation ordering; the 3B base uses it.
$L_1=\sum_{i\in S}|s_i-\hat s_i|$ measures pricing error without cancellation.}
\label{tab:error_separation}
\small
\setlength{\tabcolsep}{4pt}
\begin{tabular}{@{}llrrrr@{}}
\toprule
Base & Prices & Changes & Net pricing & $L_1$ pricing & Composition \\
\midrule
1B & midpoint & 61 & $+0.00545$ & $0.01773$ & $+0.45883$ \\
1B & exact singleton & 73 & $0.00000$ & $0.00000$ & $+0.52352$ \\
3B & midpoint & 132 & $+0.00305$ & $0.03015$ & $+1.25735$ \\
3B & exact singleton & 134 & $0.00000$ & $0.00000$ & $+1.26755$ \\
\bottomrule
\end{tabular}
\end{table}

On the two standard GPTQ banks in Table~\ref{tab:error_separation}, composition
error is $25.9\times$ the total absolute midpoint-pricing
error at 1B and $41.7\times$ at 3B. For the 3B midpoint selection, exact
singletons sum to $-1.3031$ ARC-KL, while their combination yields $-0.0457$.
For the exact-singleton selection, a base KL of $0.3695$ and a price sum of
$-1.3043$ predict a negative loss. Executing the combination yields the
$+1.2676$ correction, despite zero isolated pricing error.
KL non-negativity alone gives a lower bound of roughly $0.935$ on this
correction. Section~\ref{sec:context} tests the stronger question of whether
a monotone additive transform can represent the contextual effects.

\section{Accurate Prices for Finite Moves}
\label{sec:within}

For a primitive $d_i$, the current-state price is
$d_i^\top\nabla F(q_0)$ and its own-midpoint price is
$d_i^\top\nabla F(q_0+\tfrac12d_i)$. The latter is one-point quadrature of the
signed loss integral along the move. Midpoints lie in continuous
decoded-weight space; deployed outputs are the legal endpoints.
We use the midpoint to diagnose finite-step integration error. Exact
endpoint evaluation supplies the cheaper per-candidate control for the
backward implementation (Section~\ref{sec:constructors}).

\begin{proposition}[Midpoint identity]
\label{prop:midpoint}
If $\nabla^2F$ is $L_H$-Lipschitz on $[q_0,q_1]$ and $m=(q_0+q_1)/2$, then
$F(q_1)-F(q_0)=d^\top\nabla F(m)+r$ with
$\abs r\le L_H\norm d^3/24$. The current-state expansion has remainder
$\tfrac12 d^\top\nabla^2F(\xi)d$ for some $\xi$ on the segment,
of order $\norm d^2$.
\end{proposition}

Even Taylor terms cancel at the midpoint, making its price exact for
quadratics (proof in Appendix~\ref{app:proofs}). We measure directional accuracy on a
97-move bank containing 22 and 23 partial transitions in two Qwen3-0.6B
contexts and 26 dose batches along each of two Qwen3-4B solver steps:
level refit and code re-selection. Moves span one to $11.6$M groups and improve
the proposing reconstruction objective. Predictors use calibration blocks;
exact NLL differences use disjoint held-out blocks.

The own-midpoint response recovers $96/97$ signs with Spearman $0.9994$,
compared with $68/97$ and $0.65$ for the current-state gradient.
The proxy accepts all moves and is correct on the 54 that improve NLL.
Along the 4B code re-selection, 22 of 26 proxy-improving batches raise NLL;
the midpoint predicts all 26 directions. We estimate these responses by
central differences at the midpoint, using two forwards and checking
agreement across two step sizes.
The GPTQ banks below use hooked backward derivatives
(Appendix~\ref{app:bank}, Figure~\ref{fig:scatter}).

\paragraph{Standard GPTQ on Llama.}
Table~\ref{tab:crossarch} compares predictors on 98 Llama-3.2-1B moves:
75 layer-, family- and dose-batched transitions from RTN to GPTQ($1\%$ damping),
and 23 to a $0.1\%$-damping alternative. The INT3 g128 producer uses sequential
calibration and per-group scale/zero refits and reproduces the reference
GPTQ implementation exactly in the tested projection.
All 75 RTN$\to$GPTQ moves improve full-layer reconstruction, yet 10 raise
held-out NLL and 35 worsen ARC option-KL.
The low-damping transitions improve the group-local objective while worsening
full-layer reconstruction; 17 of 23 also raise held-out NLL.

On the same FIT units, the midpoint predicts $97/98$ NLL signs and
$96/98$ ARC signs. Its NLL mean absolute error is $30\times$ below the
current-state gradient and $39\times$ below the current-state Gauss--Newton
estimate; the corresponding ARC factors are $10\times$ and $11\times$.
Across the NLL data split, midpoint agreement is $86/98$, versus $82$ and $81$
for the two current-state estimates. The change across the split separates
data shift from finite-move integration error.

\begin{table}[!htb]
\centering
\caption{Single-move prediction on Llama-3.2-1B, standard GPTQ INT3 g128.
Entries are sign agreement counts, with Spearman in parentheses.
NLL uses 28 FIT and 28 held-out blocks; ARC option-KL uses 592
label-free items. Resolved endpoints have pointwise normal $95\%$ intervals
excluding zero (81 held-out NLL, 74 ARC), computed as the mean paired block/item
difference $\pm1.96$ standard errors. Joint agreement uses FIT NLL and ARC.}
\label{tab:crossarch}
\vspace{0.3em}
\footnotesize
\setlength{\tabcolsep}{3.5pt}
\begin{tabular}{@{}lcccccc@{}}
\toprule
& NLL, same & \multicolumn{2}{c}{NLL, held-out} & \multicolumn{2}{c}{ARC option-KL} & joint \\
\cmidrule(lr){3-4}\cmidrule(lr){5-6}
Predictor & units & all 98 & resolved & all 98 & resolved & 98 \\
\midrule
group-local reconstruction & 70 (0.01) & 71 (0.08) & 58 & 49 (0.04) & 36 & --- \\
full-layer reconstruction & 83 (0.45) & 82 (0.47) & 71 & 54 (0.08) & 39 & --- \\
current-state gradient & 85 (0.70) & 82 (0.63) & 69 & 86 (0.83) & 69 & 76 \\
current-state Gauss--Newton & 84 (0.64) & 81 (0.65) & 69 & 89 (0.82) & 73 & 78 \\
midpoint & \textbf{97} (0.9996) & \textbf{86} (0.90) & \textbf{75} & \textbf{96} (0.987) & \textbf{74} & \textbf{95} \\
\bottomrule
\end{tabular}
\end{table}

The integration pattern persists with activation ordering. On Llama-3.2-3B,
act-order GPTQ($1\%$) improves NLL over RTN by
$0.23$ $[0.17,0.27]$ on FIT units.
Across its 392 single-projection replacements, midpoint signs agree on
$390/392$ for NLL and $391/392$ for ARC, with Spearman above $0.999$.
The current gradient agrees on $153$ and $357$.
On the 1B act-order bank the corresponding counts are $223/224$ and $222/224$,
against $50$ and $179$ (Appendix~\ref{app:protocols}).

\section{Interactions Reshape the Attainable Set}
\label{sec:context}

Composition error is measurable on the state lattice itself. On a full
Boolean lattice with $f(0)=0$, the unique M\"obius expansion is
\begin{equation}
 f(x)=\sum_i s_i x_i+\sum_{i<k}\beta_{ik}x_ix_k+R_{\ge3}(x),\qquad
 \beta_{ik}=f(e_i+e_k)-f(e_i)-f(e_k).
 \label{eq:moebius}
\end{equation}
Here $s_i=f(e_i)$, and a pair coefficient $\beta_{ik}$ measures how move $k$
changes the isolated effect of move $i$. The remainder $R_{\ge3}$ contains
higher-order interactions. These exact endpoint coefficients differ from
the least-squares coefficients $\widehat\beta_{ik}$ fitted to a larger set
of states, which can absorb some higher-order effects. The four-point
measurement also underlies CLADO's cross-layer sensitivity~\citep{deng2025clado};
here we retain its discrete coefficient and measure the higher-order remainder.

\subsection{Contextual reversals and front quality}

We evaluate every state of two eight-move sublattices from a 23-primitive
Qwen3-0.6B bank. A witness lattice contains previously identified composing
pairs; a stratified lattice is selected from bank metadata alone.
Each has 256 legal states, evaluated under next-token NLL (64 blocks),
ARC-Easy option-KL (592 items) and HellaSwag option-KL (400 items).
Figure~\ref{fig:atlas} shows their interaction structure.
Fronts and joint-improvement labels in this section use exact sample means.
A prediction is false-safe when it predicts joint improvement but a measured
loss increases; intervals and resampling quantify the uncertainty separately.

\begin{figure}[!htb]
\centering
\includegraphics[width=\textwidth]{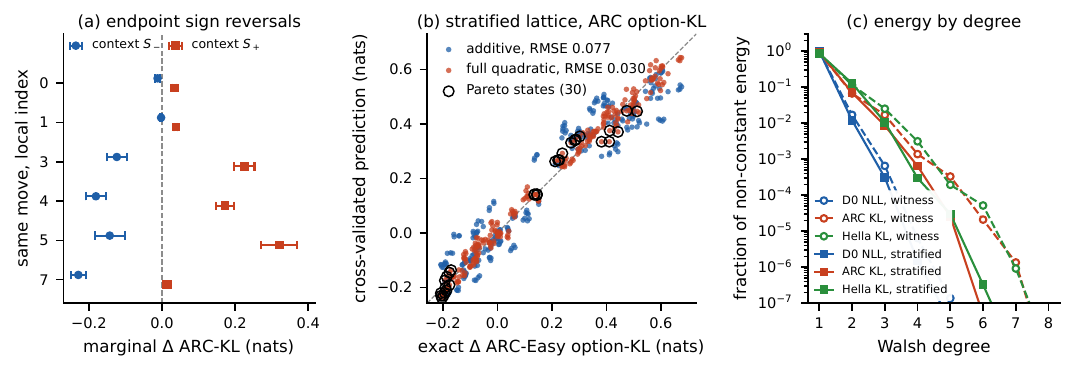}
\caption{Interactions on two exhaustive 256-state lattices.
(a) Six stratified-lattice moves reverse their ARC endpoint effect between
two contexts each. Bars are simultaneous paired-bootstrap intervals across
the searched contrasts and functionals (Appendix~\ref{app:atlas}).
(b) Stratified ARC option-KL: exact effects versus four-fold predictions;
circles mark the exact three-functional Pareto front.
(c) Walsh energy relative to non-constant energy. Small interaction components
change the direction of moves and which states are best.}
\label{fig:atlas}
\end{figure}

The same move can improve a task loss in one context and worsen it in
another. With simultaneous paired intervals, six of eight stratified-lattice
moves reverse their ARC effect, and fourteen of 28 move pairs reverse their
ordering. HellaSwag gives four sign-reversing moves and eleven ordering
reversals; NLL has no resolved same-edit sign reversal.
A monotone transform $h(a^\top x)$ cannot express either kind of reversal:
the index increment $a_i$ and ordering difference $a_i-a_k$ are independent
of context. This excludes exact representation on the measured cells.
Approximate ranking error depends on the fraction of misordered state
comparisons. Our fitted single-index control recovers $30\%$ of the
stratified front and makes 11 false-safe predictions.
These scores concern the specified fit, not the optimum over all monotone
additive indices (Appendix~\ref{app:atlas}).

Small interaction components account for these changes in ordering.
Degree-two terms carry $1.2$--$1.7\%$ of NLL's non-constant Walsh energy and
$6.7$--$12.7\%$ for the task functionals. Degree three and above contribute
$0.03$--$1.1\%$ on the stratified lattice and $0.07$--$2.9\%$ on the witness.
On the stratified lattice, a least-squares quadratic explains
$98.9$--$99.97\%$ of the variance and reduces cross-validated RMSE by
$2.5$--$5.6\times$ relative to additivity. An additive fit recovers 11 of the
30 sample-mean Pareto states, a full quadratic 26, and a rank-4 quadratic 27.
Paired resampling preserves the recall gap.

To measure the losses attainable on a predicted front, we allow a tolerance
of one median state standard error per functional. An exact-front state is
covered if an actual endpoint on the predicted front is no worse than it
by more than this tolerance in every functional. On the original stratified
sample, coverage is $66.7\%$ for addition and $100\%$ for both quadratics;
all models reach $100\%$ on the witness lattice.
Refitting the four-fold models and rebuilding the fronts in $2{,}000$ paired
resamples gives mean one-SE coverage of $70.06\%$ for addition and $99.91\%$
for the full quadratic. Their paired $95\%$ difference interval is
$[20.69,40.00]$ percentage points. On the witness lattice, mean coverage
instead favors addition, $96.14\%$ versus $91.01\%$.

Selecting one state tests a different consequence of prediction error.
The standardized worst-functional rule chooses different stratified states
from the additive and full quadratic predictions, with original-sample
regrets $0.0716$ and $0.0396$. Repeating selection within each refitted
resample leaves their regret difference unresolved, despite the front's
resolved quality-coverage gain. Appendix~\ref{app:qwen_refit} gives the paired
intervals and the matched fixed-prediction control.

\subsection{Predicting unqueried combinations}

To test larger combinations, we fit the 93 states with at most three moves.
The full quadratic recovers
$92\%$ of the Pareto front among the 163 larger subsets with zero false-safe
states, versus $52\%$ for either additive or free-intercept affine fits.
On the witness lattice, all three recover $80\%$.
A 37-query alternative computes the exact M\"obius quadratic from the base,
eight singleton and 28 pair endpoints without fitting. It recovers $90.9\%$
of the full witness front and $90.0\%$ of the full stratified front.
Appendix~\ref{app:atlas} compares the estimators and fixed-utility regret.

\paragraph{Replication on standard GPTQ.}
The low-order composition structure also appears on Llama-3.2-3B with
activation-ordered INT3 GPTQ. Eight projection replacements chosen from
architecture metadata form a new lattice, evaluated on the original 3B FIT
units. Its 37 base, singleton and pair
states determine a M\"obius quadratic that predicts the remaining 219
combinations (Table~\ref{tab:independent_l8_main}).
The simultaneous reversal audit resolves no opposite-sign witnesses on
this bank (Appendix~\ref{app:independent_l8}).

\begin{table}[!htb]
\centering
\caption{Standard-GPTQ lattice: prediction on 219 unqueried combinations.
Recall and precision concern the sample-mean Pareto front restricted to those
219 states. Both models use the same exact isolated effects; the quadratic
adds 28 pair interactions.}
\label{tab:independent_l8_main}
\footnotesize
\setlength{\tabcolsep}{4pt}
\begin{tabular}{@{}lrrrrr@{}}
\toprule
Predictor & NLL MAE $\downarrow$ & ARC-KL MAE $\downarrow$ & Recall $\uparrow$ & Precision $\uparrow$ & False-safe $\downarrow$ \\
\midrule
Singleton additive & $1.467\!\times\!10^{-4}$ & $4.110\!\times\!10^{-3}$ & 5/6 & 5/6 & 1 \\
Pair quadratic & $1.835\!\times\!10^{-5}$ & $7.843\!\times\!10^{-4}$ & 6/6 & 6/6 & 0 \\
\bottomrule
\end{tabular}
\end{table}

Pair information reduces NLL MAE by $8.0\times$ and ARC-KL MAE by $5.2\times$. For the
three-move mask 137, the ARC-KL prediction
changes from $-0.00224$ to $+0.00326$, close to the actual $+0.00376$
(paired interval $[-0.00288,+0.01054]$). Its pair correction is positive with
interval $[+0.00370,+0.00736]$: the interaction is resolved, while the composed
endpoint direction is not (Appendix~\ref{app:l8_example}).

Front membership and loss-space quality give different assessments here.
In 2,000 paired bootstrap draws of the original FIT units, we recompute every
anchor and predictor. Within the 219 unqueried combinations, mean front
recall is $97.42\%$ for the anchored quadratic and $79.33\%$ for exact-singleton
addition; mean false-safe counts are $0.108$ and $1.560$.
At one fixed pooled standard error per functional, however, the actual
endpoints returned by the predicted fronts cover a mean $100.00\%$ and
$99.51\%$ of exact-front quality. The membership gap largely separates
states with nearly equal losses.

For a single decision with fixed normalized NLL-plus-ARC utility, both models
select mask 133 on the original sample. Refit and re-selection, with both
candidates and the exact oracle restricted to unqueried states, give mean
exact regrets $0.000365$ and $0.06490$. Pair information improves selection
stability across these resamples even though the original choices coincide.
Appendix~\ref{app:l8_refit} gives the full-domain comparison and numerical scales.
Full pair anchoring requires $1+n+\binom n2$ endpoint queries for $n$
compatible moves; its predictive value here is measured on eight-move sublattices.

\subsection{Functional-specific ordering}
\label{sec:functional}

Held-out NLL and ARC option-KL disagree in direction on $46/98$ standard-GPTQ
moves, including $31/62$ whose two endpoint effects are resolved.
Midpoint responses classify conflict on $95/98$ FIT comparisons and
$59/62$ resolved comparisons across the split.
The act-order banks give $81/224$ conflicts at 1B and $183/392$ at 3B;
midpoint predicts $389/392$ conflict labels on the latter, versus $164$ for the
current gradient.

A corresponding 4B example improves held-out NLL by $0.00259$ while lowering
MMLU and ARC-Challenge gold margins by $0.054$ and $0.065$.
Appendices~\ref{app:structure} and~\ref{app:unprotected} give the cross-domain
response matrix and the dependence on the declared functional set.

\section{Endpoint Structure and Shared-Center Prices}
\label{sec:centers}

A shared backward pass offers a way to obtain many prices at once.
At the common proposal midpoint, averaging the per-unit responses gives
$v_i=d_i^\top\nabla F(q_0+\tfrac12D\onev)$.
This evaluates each move with every proposed partner present at half strength.
A subset's own midpoint omits its unselected partners.

\paragraph{What endpoints determine.}
\label{sec:probes}
The low-order structure in Section~\ref{sec:context} describes legal binary
states. Adding $x_i(1-x_i)h(x)$ to a smooth extension of $f$ leaves all binary
endpoints unchanged and can change interior gradients. Endpoint measurements
alone therefore do not determine shared-center prices. Constant curvature
links the two and supplies a reference against which to test the measured
gradient responses.

\begin{proposition}[Coefficients and common-center pricing]
\label{prop:quadratic}
If $\nabla^2F\equiv H$ on the lattice's convex hull, then $R_{\ge3}\equiv0$ and
\begin{equation}
 s_i=d_i^\top\nabla F(q_0+\tfrac12d_i)
     =d_i^\top\nabla F(q_0)+\tfrac12d_i^\top Hd_i,\qquad
 \beta_{ik}=d_i^\top Hd_k.
 \label{eq:coeffs}
\end{equation}
The common-center price satisfies $v_i=s_i+\tfrac12\sum_{k\ne i}\beta_{ik}$,
so, with $\beta$ symmetric and zero-diagonal,
\begin{equation}
 \sum_i v_i x_i-f(x)=\tfrac12 x^\top\beta(\onev-x).
 \label{eq:drift}
\end{equation}
\end{proposition}

Here common-center pricing splits each pair interaction equally between its
participants, the Shapley allocation of a quadratic game. A subset's price
counts its internal interactions correctly and retains half its interaction
with omitted moves. Adding $i$ after a realized set $S$ instead costs
$s_i+\sum_{k\in S}\beta_{ki}$, which sequential own-midpoint measurements
recover (Appendix~\ref{app:proofs}).

\paragraph{Testing the shared-center model.}
For $r(z)=D^\top\nabla F(q_0+\tfrac12Dz)$, constant curvature gives
$2[r(z)-r(0)]=Bz$ with $B=D^\top HD$.
Unlike the exact endpoint $\beta$ and fitted $\widehat\beta$, $B$ describes
continuous gradient responses. On the Qwen3 lattices, an operator recovered
at the eight unit centers mispredicts responses at other centers on
$31$--$39$ of 40 components, even though endpoint quadratics predict well.
The endpoint half-cut in Equation~\eqref{eq:drift} also poorly predicts the
measured common-center error: relative RMSE is $0.99$--$1.92$ across the six
combinations of lattice and functional. Appendix~\ref{app:cut_decomposition} separates this
discrepancy into higher-order endpoint terms and the mismatch between
common-center responses and endpoint pair allocation.

\paragraph{Measured center drift.}
\label{sec:drift}
For varying curvature, let $d_S,d_{\bar S}$ be selected and omitted
displacements, $m_S=q_0+\tfrac12d_S$, and $m=m_S+\tfrac12d_{\bar S}$.
Proposal-center minus own-midpoint response is exactly
$\tfrac12\int_0^1d_S^\top\nabla^2F(m-\tfrac t2d_{\bar S})d_{\bar S}\,dt$;
the gap to an endpoint also contains the midpoint remainder.

On a Qwen3-1.7B bank, three frozen candidates and three functionals give nine
paired center comparisons on identical units. All nine drift intervals
exclude zero and seven reverse the response sign
(Table~\ref{tab:driftfull}). For support four, drift is
$+0.301$ $[+0.228,+0.374]$ in NLL and $-0.073$ $[-0.080,-0.066]$ in target KL.
Moving from the candidate midpoint to the proposal midpoint shifts the two
functional responses in opposite directions.

\section{What Changes in State Selection?}
\label{sec:constructors}

The lattices establish what pair information adds to prediction. On larger
categorical banks, we compare restricting the number of moves with querying
their combinations directly.
An \emph{additive selector} chooses a compatible set of replacements minimizing
summed ARC-KL prices subject to predicted NLL non-increase.
An \emph{endpoint beam} retains non-dominated states and evaluates their
extensions directly, using combination endpoints rather than fitted pair
coefficients. One own-midpoint backward costs $1.5$--$3.6$ endpoint forwards
on these banks (Appendix~\ref{app:compute}), so exact endpoints provide the
cheaper per-candidate prices for these selection comparisons.

\paragraph{Large-support composition on standard GPTQ.}
On Llama-3.2-1B without activation ordering, midpoint prices replace 61 of
112 projections and reduce test perplexity from $24.67$ to $21.94$, but leave
ARC-KL unresolved. Exact-singleton additive selection makes 73 changes and
increases test ARC-KL by $0.044$ $[0.017,0.071]$.
A fixed nested path through the 61 changes improves both FIT losses at
support 8 and 16, then worsens ARC at support 32. Its ARC composition
residual grows from $0.039$ at support 8 to $0.391$ at support 48
(Appendix~\ref{app:support_path}).

On the stronger 3B activation-ordered base, the midpoint, current and
exact-singleton selectors make 132, 149 and 134 changes. Executing these
previously selected states raises the original TEST NLL by $0.064$, $0.094$ and $0.063$,
respectively, with all three intervals excluding zero. Their ARC intervals
span zero. The midpoint state's ARC error is dominated by composition,
whereas its NLL prediction error contains a larger data-shift term
(Appendix~\ref{app:3b_composed}).

\paragraph{Support-matched selection.}
On the 1B bank without activation ordering, the endpoint beam's 352-forward
budget produces a two-change state with $19\%$ lower test perplexity.
To separate the effect of support
size from state selection, we cap exact-singleton addition at two, four and
eight changes. Two is the primary matched control; four and eight test
sensitivity to the cap. The objective, complete candidate encodings and
validation constraint remain unchanged. All three pass validation.
Table~\ref{tab:support_matched} uses the original evaluation units for this
retrospective comparison; the beam state is fixed rather than re-searched.

\begin{table}[!htb]
\centering
\caption{Restricting support already improves the additive selector.
Llama-3.2-1B, INT3 GPTQ without activation ordering; TEST differences from
the base on 114 NLL blocks and 297 ARC-Easy items, paired $95\%$ intervals.
All caps use 224 exact singleton prices; the beam uses 352 endpoint queries.}
\label{tab:support_matched}
\footnotesize
\setlength{\tabcolsep}{3.5pt}
\begin{tabular}{@{}lccc@{}}
\toprule
Selector & Changes & $\Delta$ NLL & $\Delta$ ARC option-KL \\
\midrule
Additive, cap 2 & 2 & $-0.111$ $[-0.174,-0.060]$ & $-0.038$ $[-0.050,-0.025]$ \\
Additive, cap 4 & 4 & $-0.112$ $[-0.174,-0.060]$ & $-0.044$ $[-0.058,-0.029]$ \\
Additive, cap 8 & 8 & $-0.172$ $[-0.288,-0.065]$ & $-0.043$ $[-0.060,-0.026]$ \\
Endpoint beam & 2 & $-0.206$ $[-0.260,-0.156]$ & $-0.029$ $[-0.040,-0.017]$ \\
\bottomrule
\end{tabular}
\end{table}

Making fewer changes suffices to improve both losses. At the same support
of two, the beam has a $0.0948$ NLL advantage over addition
($95\%$ interval $[0.0520,0.1354]$), while their ARC difference is unresolved.
Cap four instead improves ARC over the beam by $0.0147$
$[0.0013,0.0278]$ and sacrifices NLL.
Both pass the NLL constraint, so cap four favors the declared ARC objective
at lower query cost, while the beam achieves a larger NLL decrease.
On the 1B act-order base, the beam's small effects relative to the base remain
unresolved, whereas the 37-change midpoint state raises NLL
(Table~\ref{tab:llama_oneshot}).

\paragraph{Support control on new 3B data.}
On the 3B act-order bank, the support-two beam uses 588 charged queries:
392 archived singletons and 196 new pairs extending one singleton.
The additive controls use the same singleton prices, with caps two, four
and eight. FIT fixes all candidates before NLL evaluation on 256 new
WikiText blocks from 123 articles (Appendix~\ref{app:3b_support_nll}).
The 134-change additive state raises NLL by $0.06643$
$[0.04696,0.09352]$; all three capped-additive intervals span zero.
The two-change beam lowers NLL by $0.00297$ $[0.00041,0.00604]$ relative to
the base. Its paired difference from cap-two addition is $-0.00189$
$[-0.00417,+0.00006]$, leaving the additional gain unresolved.
On FIT, cap-two addition has lower ARC-KL and the beam has lower NLL.

\paragraph{Re-evaluation after execution.}
Changing the center changes which alternatives are attractive.
After 16 changes, the ARC price reverses for $117/192$ alternatives whose
own projection is unchanged (Figure~\ref{fig:reversal}).

\begin{figure}[!htb]
\centering
\includegraphics[width=0.72\textwidth]{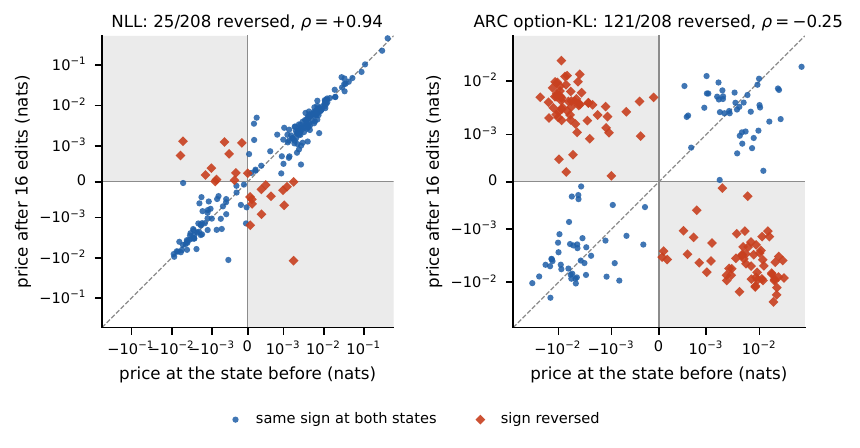}
\caption{Own-midpoint prices before and after 16 projection changes on
Llama-3.2-1B. Shaded quadrants mark sign reversals.
ARC prices reverse on 121/208 comparable alternatives; restricting to the
192 whose own projection is unchanged gives 117 reversals.
NLL retains rank correlation $0.94$, with 25/208 sign reversals.
Appendix~\ref{app:protocols} gives both consecutive-center comparisons.}
\label{fig:reversal}
\end{figure}

With batches of 16 moves, recentering improves both functionals in the first
development round. The validation-selected state later improves fresh-test
ARC-KL by $0.046$ and worsens NLL by $0.090$ relative to its start.
Re-pricing at the reached state accounts for the moves already executed;
interactions within the next batch still affect their combined utility.

\paragraph{When isolated selection is enough.}
On the Qwen3-4B reconstruction-optimal code re-selection bank, 26 exact
singleton queries select the same nine-batch state as midpoint pricing at
94 forward-equivalents. It improves both losses over the full solver output:
$-0.057$ NLL and $-0.166$ ARC-KL, with paired intervals excluding zero.
A 199-query direct endpoint beam improves ARC over a screened beam, but does
not resolve an advantage over the additive state. The screened beam's
midpoint rows predict $72/72$ signs versus $32/72$ for its current rows,
but zeroing or negating them changes no search decision.
Full costs, comparisons and the differing lattice structures of level refit,
code re-selection and grid snapping appear in
Appendices~\ref{app:bridge} and~\ref{app:protocols}.

The resulting gains are functional-specific. The 4B additive state gains
$2.7$--$2.9$ accuracy points under ARC, MMLU and HellaSwag label choice,
but standard HellaSwag completion gives no resolved improvement for any
repair (Table~\ref{tab:constructors}).

\section{Related Work}
\label{sec:related}

Reconstruction objectives underlie AdaRound, BRECQ, GPTQ, AWQ and
OmniQuant~\citep{nagel2020adaround,li2021brecq,frantar2023gptq,lin2024awq,shao2024omniquant}.
Loss-aware quantization incorporates end-task information into this choice.
GuidedQuant~\citep{kim2025guidedquant} minimizes gradient-weighted output
reconstruction error, corresponding to a block-Fisher quadratic approximation
that retains within-output-channel interactions.
DiscQuant~\citep{chee2025discquant} uses linearized end-loss constraints in
rounding. Our current-gradient and Gauss--Newton predictors isolate particular
approximations on a fixed bank of finite moves.
Cross-block interactions and high-loss grid choices have also been studied
directly~\citep{shabanovi2024interactions,jhang2026qat_basin}.

PQI~\citep{hu2025pqi} observes that exact isolated layer losses do not add
when layers are quantized together. It integrates between full-precision and
quantized weights to improve sensitivity estimates.
QIG~\citep{xiang2026qig} applies integrated gradients to token sensitivity.
Our midpoint price is one-point quadrature of the same path identity.
Using legal endpoints at the same bitrate, we isolate pricing from
composition on identical units, measure the error that remains with exact
singleton prices, and test how combination measurements inform prediction
and selection.

CLADO~\citep{deng2025clado} shows that isolated loss measurements can select
the wrong layer pair. Its four-point endpoint difference is the same observable
as $\beta$ in Equation~\eqref{eq:moebius}. CLADO interprets it as a Hessian
bilinear approximation and uses a positive-semidefinite sensitivity matrix
for mixed-precision integer quadratic programming. CoopQ~\citep{zhao2025coopq}
uses Shapley marginal-contribution covariance for quadratic allocation.
Our fixed-rate lattices measure the exact discrete pair coefficients and
higher-order remainder, then test predictions on unqueried combinations.
We assess the resulting front by its state membership, attained losses and
the utility of a single selected state.

Boolean quantization configurations have also been studied through
Walsh/M\"obius decompositions and monotone additive
indices~\citep{hill2026coverage}. That study reports contextual ordering
reversals alongside close approximate ranking agreement, with a small
fraction of misordered configuration pairs in the measured settings.
Our witnesses test exact representation for fixed-rate replacements of
already quantized states. The held-out subset predictions and quality-tolerant
front comparisons then measure the usefulness of particular additive and
pairwise predictors in this different intervention space.
Cross-layer compensation~\citep{noda2026crosslayer}, task-aware
allocation~\citep{guo2025tacq,chen2025taq}, and multi-objective gradient
methods~\citep{sener2018multitask,yu2020pcgrad} address related dependencies.
Our controls compare what exact isolated effects, pair measurements and
direct combination queries contribute to state selection.

\paragraph{Scope.}
The structural tests cover two Qwen sublattices sharing five moves and a
metadata-selected Llama lattice. The construction comparisons span different
producer banks and declared functionals.

\paragraph{Conclusion.}
\label{sec:conclusion}
Exact isolated effects leave a measurable error when quantization moves are
composed. Same-data controls separate this error from pricing, and resolved
Qwen reversals exclude an exact monotone additive representation.
Pair measurements improve prediction of unqueried combinations, with
different consequences for selection. The stratified Qwen front retains a
quality-coverage advantage after refitting, while its single-state regret
difference remains unresolved. On the 3B GPTQ lattice, the additive front
remains close in loss space despite a gap in membership recovery.
Endpoint structure also leaves continuous gradient responses undetermined,
so shared-center prices require their own assessment. On larger banks,
support limits curb the failures of large combinations, and extra endpoint queries
change the NLL/ARC trade-off in a bank-dependent way.

\label{endofmaintext}

\section*{Reproducibility Statement}
Appendices~\ref{app:setup}--\ref{app:protocols} specify the bank construction,
statistical units, data-use history and selection rules behind each result.
Appendix~\ref{app:proofs} gives the mathematical proofs.
The GPTQ implementation is compared with the reference at the decoded-weight
level in the configurations used.
The accompanying source and measurement records support paired comparisons
and figure reproduction; the figure scripts recompute the quantities they
display from the stored measurements.

\section*{AI Use Statement}
Generative AI assistants supported manuscript drafting, code implementation,
and experiment execution and monitoring. We reviewed the resulting text,
claims, code and measurements, including the implementation comparisons and
source records underlying reported results, and take responsibility for the
content of this work.

\bibliographystyle{iclr2027_conference}
\bibliography{references}

\appendix
\section{Proofs}
\label{app:proofs}

\subsection{Midpoint identity (Proposition~\ref{prop:midpoint})}

\begin{proof}
Let $u=d/2$.  Second-order Taylor expansion with integral remainder about $m$
gives $F(m\pm u)=F(m)\pm\nabla F(m)^\top u+\tfrac12u^\top\nabla^2F(m)u+R_\pm$
with
\begin{equation}
  \abs{R_\pm}
  =\Bigl|\int_0^1(1-s)\,u^\top[\nabla^2F(m\pm su)-\nabla^2F(m)]u\,ds\Bigr|
  \le\int_0^1(1-s)\,L_Hs\norm u^3ds=\frac{L_H\norm u^3}{6}.
\end{equation}
Subtracting, the even terms cancel:
$F(q_1)-F(q_0)=2\nabla F(m)^\top u+(R_+-R_-)=d^\top\nabla F(m)+r$ with
$\abs r\le 2L_H\norm{d/2}^3/6=L_H\norm d^3/24$.  For the current-state
expansion, Taylor's theorem with Lagrange remainder gives
$F(q_1)-F(q_0)=d^\top\nabla F(q_0)+\tfrac12d^\top\nabla^2F(\xi)d$ for some $\xi$
on the segment.
\end{proof}

A Taylor expansion about the floating reference $W$ instead controls the
paths $[W,q_0]$ and $[W,q_1]$. If the same Hessian Lipschitz bound $L_H$ holds
on both paths, subtracting their quadratic expansions gives a remainder bound
$\frac{L_H}{6}(\norm{a-d/2}^3+\norm{a+d/2}^3)$, where $a=m-W$.
For $d\ne0$ and $L_H>0$, its ratio to the midpoint bound is at least
$8(\norm a/\norm d-\tfrac12)_+^3$. This comparison concerns upper bounds,
not the realized errors: the reference-centered bound can grow arbitrarily
larger as the move shrinks relative to the total quantization displacement.

\subsection{Coefficients under constant curvature (Proposition~\ref{prop:quadratic})}

\begin{proof}
With $\nabla^2F\equiv H$ on the convex hull of the lattice,
$F(q_0+Dx)-F(q_0)=x^\top D^\top\nabla F(q_0)+\tfrac12x^\top D^\top HDx$ exactly.
Write $a_i=d_i^\top\nabla F(q_0)$ and $B=D^\top HD$.  Since $x_i^2=x_i$ on
$\{0,1\}^n$,
$\tfrac12x^\top Bx=\tfrac12\sum_iB_{ii}x_i+\sum_{i<k}B_{ik}x_ix_k$, hence
$f(x)=\sum_i(a_i+\tfrac12B_{ii})x_i+\sum_{i<k}B_{ik}x_ix_k$ and $R_{\ge3}\equiv0$.
Comparing with the Möbius expansion \eqref{eq:moebius} gives
$s_i=a_i+\tfrac12B_{ii}$ and $\beta_{ik}=B_{ik}=d_i^\top Hd_k$.  The own-midpoint
response is $d_i^\top\nabla F(q_0+\tfrac12d_i)=a_i+\tfrac12d_i^\top Hd_i=s_i$.
At the full-proposal midpoint $m=q_0+\tfrac12D\onev$,
$v_i=d_i^\top\nabla F(m)=a_i+\tfrac12\sum_kB_{ik}=s_i+\tfrac12\sum_{k\ne i}\beta_{ik}$.
Summing over a subset,
$\sum_iv_ix_i=\sum_is_ix_i+\tfrac12\sum_i\sum_{k\ne i}\beta_{ik}x_i$, and
subtracting $f(x)=\sum_is_ix_i+\sum_{i<k}\beta_{ik}x_ix_k
=\sum_is_ix_i+\tfrac12\sum_i\sum_{k\ne i}\beta_{ik}x_ix_k$ leaves
$\tfrac12\sum_i\sum_{k\ne i}\beta_{ik}x_i(1-x_k)=\tfrac12x^\top\beta(\onev-x)$.
\end{proof}

Without constant curvature, for a subset $S$ with displacement $d_S$ and
complement $d_{\bar S}$, the fundamental theorem of calculus along the segment
from the own midpoint $m_S=q_0+\tfrac12d_S$ to the proposal midpoint
$m=m_S+\tfrac12d_{\bar S}$ gives
\begin{equation}
  \sum_{i\in S}v_i-d_S^\top\nabla F(m_S)
  =\tfrac12\int_0^1d_S^\top\nabla^2F\!\bigl(m-\tfrac t2d_{\bar S}\bigr)d_{\bar S}\,dt ,
\end{equation}
which reduces to $\tfrac12x^\top\beta(\onev-x)$ when the Hessian is constant.
To compare with the endpoint effect, write
$f(S)=d_S^\top\nabla F(m_S)+r_S$. The common-center error is therefore the
integral above minus $r_S$, with $\abs{r_S}\le L_{H,S}\norm{d_S}^3/24$ under
the assumptions of Proposition~\ref{prop:midpoint} on $[q_0,q_0+d_S]$.

\subsection{Sequential midpoint pricing}

\begin{proposition}[Sequential midpoint pricing]
\label{prop:path}
For a path $q_t=q_{t-1}+d_t$ with midpoints $m_t$,
$F(q_T)-F(q_0)=\sum_t d_t^\top\nabla F(m_t)+\sum_t r_t$ with
$\abs{r_t}\le L_{H,t}\norm{d_t}^3/24$, where $L_{H,t}$ bounds the Hessian's
Lipschitz constant on segment $t$; under constant curvature every $r_t$
vanishes, and within a fixed binary chart $q(S)=q_0+\sum_{i\in S}d_i$ the exact
marginal of adding move $i$ to an already selected set $S$ is
$s_i+\sum_{k\in S}\beta_{ki}$.  (For categorical replacements the constructors
of Section~\ref{sec:constructors} use the pathwise identity with state-specific
moves.)
\end{proposition}

\begin{proof}
Apply Proposition~\ref{prop:midpoint} to each segment $[q_{t-1},q_t]$ and sum;
the telescoping left-hand sides give $F(q_T)-F(q_0)$.  Under constant curvature
$H$ each $r_t$ is zero.  For the marginal, with $q_S=q_0+\sum_{k\in S}d_k$ and
$m=q_S+\tfrac12d_i$, $d_i^\top\nabla F(m)=d_i^\top\nabla F(q_0)+\sum_{k\in S}d_i^\top Hd_k+\tfrac12d_i^\top Hd_i
=s_i+\sum_{k\in S}\beta_{ki}$ by \eqref{eq:coeffs}.
\end{proof}

\subsection{Undeclared functionals}
\label{app:unprotected}

Let an algorithm choose $x\ne0$ from the responses of the declared functionals
alone.  For any undeclared functional pick a response row $w_+$ with
$w_+^\top x>0$ and set $w_-=-w_+$; the two worlds present identical declared
inputs, hence identical output $x$, and opposite signs for the undeclared
change.  No function of the declared responses constrains that sign; measuring
the undeclared functional supplies additional evidence, as the D5 prediction in
the ARC construction of Appendix~\ref{app:arc_history} illustrates.

\section{Experimental Objects and Functionals}
\label{app:setup}

\subsection{Models, recovered states and storage}

Table~\ref{tab:state_definitions} identifies the Qwen inputs used in the move
banks. A floating master supplies the non-quantized tensors and, unless noted
otherwise, the reference distribution for KL. An assignment transplant keeps
the recipient's signs and non-projection tensors, imports the donor's
strong/weak assignments, and refits both magnitudes against the recipient's
weights and reconstruction weights. Thus $W_A$ and $W_C$ name floating
recipients, whereas $z_A$ and $z_C$ name assignment planes.

\begin{table}[!htb]
\centering
\caption{Input-state definitions. The two 4B families have separate origins:
the level-refit and code re-selection ladder acts on a grid-trained master;
grid snapping acts on a separately recovered free-codebook state.}
\label{tab:state_definitions}
\vspace{0.3em}
\small
\setlength{\tabcolsep}{4pt}
\begin{tabular}{@{}>{\raggedright\arraybackslash}p{0.20\textwidth}>{\raggedright\arraybackslash}p{0.76\textwidth}@{}}
\toprule
Object & State definition \\
\midrule
0.6B $W_A,W_C$ & Qwen3-0.6B, 1600-step classic two-bit recovery on the same
teacher-generated token plan. $W_A$ uses ordinary reconstruction;
$W_C$ uses online optimizer-weighted reconstruction. Both use forward-KL
distillation to a fixed teacher, floating-point initialization, learning rate $2\times10^{-5}$,
batch size 8, sequence length 512 and seed 20260815. The four transplanted
states pair recipients and assignments as $(W_A,z_A)$, $(W_A,z_C)$,
$(W_C,z_A)$ and $(W_C,z_C)$. \\
\addlinespace
4B ladder A/B/C & Qwen3-4B grid-trained floating master at recovery step 6,104.
A is its grid-constrained hard state;
B releases the levels with codes fixed; C exhaustively re-selects
magnitude-ranked codes and refits free levels. \\
\addlinespace
4B classic A & Qwen3-4B after 1600-step free-codebook recovery,
with baseline PPL $20.929308$.
Grid snapping projects its levels onto the target grid while preserving
the original code planes.
The Lloyd-reassigned state is a separate control. For this lattice, option-KL
uses the classic A state itself as the fixed reference. \\
\addlinespace
1.7B center audit & Original Qwen3-1.7B floating model, without recovery.
The starting state assigns the largest half of each group's magnitudes
to its strong level; the proposed state optimizes the strong count and
both levels under uniform reconstruction weights. The changed, proxy-improving
groups form the 25-primitive bank. \\
\bottomrule
\end{tabular}
\end{table}

Appendices~\ref{app:bank}--\ref{app:protocols} give the bank construction,
solver-step transformations and execution protocols, including the Llama
standard-GPTQ producer and its activation-ordering variants.

\paragraph{Same-rate encoding.}
Each Qwen projection weight has one sign bit and one strong/weak bit.
In the scientific measurement artifacts, each row-wise group of 128 weights
stores two FP32 magnitudes: the decoder payload is therefore
$2+64/128=2.5$ bits per quantized weight. Free and grid-constrained states use
this same encoding. The grid restricts each magnitude to an E2M1 value times
a UE4M3 group scale and a fixed tensor scale; it does not change the allocated
two-magnitude payload. This is distinct from native NVFP4 block-16 storage.
Per-group reconstruction objectives and selected-count records in the research
bundles are auxiliary measurements, not decoder inputs. Llama INT3 artifacts
store 3-bit codes, one FP32 scale and one INT16 zero point per group, giving
$3+48/128=3.375$ bits per quantized weight. These rates exclude shared
non-projection tensors and file-container overhead. A replacement takes codes
and all decoding metadata together; tensor shapes and non-projection tensors
remain fixed within each comparison.

\subsection{Losses and finite-move predictors}

Let $p_q(\cdot\mid x_{<t})$ be the model's next-token distribution.
NLL is the mean of $-\log p_q(x_t\mid x_{<t})$ over the declared target tokens.
Full-vocabulary KL is
$\mathrm{KL}(p_{\mathrm{ref}}\Vert p_q)$ at the same prefixes, averaged over
those tokens. All logs are natural; scoring uses temperature 1.

For multiple choice, let $\ell_{q,a}(x)$ be the score of valid option $a$ and
define $\pi_{q,a}(x)=\exp\ell_{q,a}(x)/\sum_b\exp\ell_{q,b}(x)$.
The target functional is
\begin{equation}
 F_{\mathrm{opt}}(q)=\frac1N\sum_{x=1}^{N}\sum_a
 \pi_{\mathrm{ref},a}(x)
 \log\frac{\pi_{\mathrm{ref},a}(x)}{\pi_{q,a}(x)}.
 \label{eq:option_kl_definition}
\end{equation}
The Qwen response experiments use the next-token logits of the answer labels
after a direct-answer chat prompt with thinking disabled.
The Llama response experiments score each complete answer continuation:
$\ell_{q,a}=\sum_t\log p_q(a_t\mid x,a_{<t})$, without length normalization.
Prompt tokens include the tokenizer's special tokens; answer continuations
are tokenized separately without them. Both protocols normalize over valid
options with no temperature rescaling and hold the reference probabilities
fixed across the bank. Labels are unnecessary for this KL. When evaluated,
option cross-entropy is $-\log\pi_{q,y}$ and gold margin is
$\ell_{q,y}-\max_{a\ne y}\ell_{q,a}$; accuracy uses $\arg\max_a\ell_{q,a}$.
The separate standard HellaSwag completion control reports raw and
length-normalized accuracy (Appendix~\ref{app:protocols}).

\paragraph{Current-state Gauss--Newton.}
Let $z(q)$ denote vocabulary logits for NLL or the option scores $\ell(q)$ for
option-KL, $u=J_z(q_0)d$, and $p=\operatorname{softmax}z(q_0)$.
The implemented predictor is
\begin{equation}
 d^\top\nabla F(q_0)
 +\frac12\left\langle
 u^\top[\operatorname{diag}(p)-pp^\top]u
 \right\rangle,
 \label{eq:ggn_definition}
\end{equation}
where the average uses the functional's token or item weights.
The Jacobian-vector product is evaluated at the current state. The quadratic
term retains curvature of the loss with respect to $z$, while omitting the
second derivatives of the model's logits or sequence scores.

\paragraph{Reconstruction proxy.}
\label{app:reconstruction}
For $R(q)=\tfrac12(q-W)^\top\Lambda(q-W)$, the exact proxy difference is
$R(q_1)-R(q_0)=d^\top\Lambda(m-W)$.
Here $W$ is the recipient floating master and $\Lambda$ is its fixed diagonal
reconstruction weight matrix, distinct from the displacement matrix $D$ in
Section~\ref{sec:setup}. The Qwen weighted refits use
$h_i=\max(\texttt{exp\_avg\_sq}_i,10^{-12})$ from that recipient's optimizer
snapshot. These weights enter the reconstruction objective and never the
functional definition. The uniform 1.7B bank uses $h_i=1$.

For the standard-GPTQ audit, let $E_p=Q_p-W_p$ be the decoded error of
projection $p$ relative to its floating weight matrix. The common covariance is
\[
 H_p=\frac{2}{N_{\mathrm{seq}}}\sum_{b=1}^{N_{\mathrm{seq}}}X_{p,b}^\top X_{p,b},
 \qquad N_{\mathrm{seq}}=128,
\]
where $X_{p,b}$ contains all 2,048 token-input rows for sequence $b$, captured
in the full GPTQ($1\%$) model. This $H_p$ is frozen across RTN and both GPTQ
states; no damping is added to the audit covariance. The two scores are
\[
 R_{\mathrm{full}}(Q)=\sum_p\operatorname{tr}(E_pH_pE_p^\top),\qquad
 R_{\mathrm{group}}(Q)=\sum_{p,r,g}e_{p,r,g}^\top H_{p,gg}e_{p,r,g},
\]
where $e_{p,r,g}$ is the column vector of errors in row $r$ and contiguous
128-column group $g$. The group-local score drops cross-group terms.
Each move's score is $R(Q_{\mathrm{right}})-R(Q_{\mathrm{left}})$ in its actual
left context, including the prefix state for a dose chain. These fixed-context
audit scores are distinct from the changing layer inputs and damped matrices
used during sequential GPTQ construction.

\subsection{Data domains and statistical units}

Table~\ref{tab:domain_definitions} defines the six-domain Qwen response matrix.
Token-weighted domains divide summed losses by the total number of predicted
tokens. Document-weighted domains first average within each document, then
average documents equally. A block supplies its own causal context; padding
tokens are excluded.

\begin{table}[!htb]
\centering
\caption{The six Qwen data domains. Counts refer to the full response matrix;
later experiments select the block or item partitions stated in their protocols.}
\label{tab:domain_definitions}
\vspace{0.3em}
\small
\setlength{\tabcolsep}{4pt}
\begin{tabular}{@{}l>{\raggedright\arraybackslash}p{0.58\textwidth}r@{}}
\toprule
Domain & Source and selection & Units / targets \\
\midrule
D0 & \texttt{Salesforce/wikitext}, \texttt{wikitext-103-raw-v1}, validation.
Fixed token-stream prefix; 512 targets per block. & 128 / 65,536 \\
\addlinespace
D1 & Fixed teacher-generated evaluation stream, seed 20260819;
512 targets per block. & 128 / 65,536 \\
\addlinespace
D2 & Fixed C4 confirmation records from the English-domain, long-document
and language-modeling subsets, preserving record order and boundaries.
& 56 / 39,087 \\
\addlinespace
D3 & \texttt{fla-hub/pg19}, validation record 32, fixed DEV token segment;
blocks of at most 512 targets. & 77 / 39,087 \\
\addlinespace
D4 & \texttt{launch/gov\_report}, GAO test documents. Frozen document/window
hash ordering; one 2,444-token window per DEV document. & 8 / 19,544 \\
\addlinespace
D5 & \texttt{coastalcph/lex\_glue}, SCOTUS test opinions. Frozen
document/window hash ordering; one 2,444-token window per DEV document.
& 8 / 19,544 \\
\bottomrule
\end{tabular}
\end{table}

D0--D3 are token weighted; D4--D5 are document weighted. D3--D5 denote their
DEV subsets throughout the Qwen experiments. Record and document selections
remain fixed across model states. D0 and D1 stream blocks share one boundary token and contain 512
predictions each. The 0.6B atlas instead uses all 64 non-overlapping 2,048-token
input blocks of the fixed D0 array, with 2,047 predictions per block.
The 97-move experiment splits D2 into even prediction blocks and odd endpoint
blocks. The 1.7B center audit uses residue 2 modulo 4 of D0/D1, giving 32
blocks per domain, together with 800 CommonsenseQA training items in 100
microbatches. The 4B repair repartitions previously consumed D2/ARC units;
its exact partition history is given in Appendix~\ref{app:protocols}.

\subsection{What each evaluation holds out}

Table~\ref{tab:evidence_protocols} separates prediction of unseen combinations
on existing units from evaluation on data excluded from a recorded history.

\begin{table}[t]
\centering
\caption{Data-use scope of the main composition and selection comparisons.
FIT denotes the units used to measure prices or fit predictors.}
\label{tab:evidence_protocols}
\small
\setlength{\tabcolsep}{4pt}
\begin{tabular}{@{}>{\raggedright\arraybackslash}p{0.19\textwidth}>{\raggedright\arraybackslash}p{0.36\textwidth}>{\raggedright\arraybackslash}p{0.40\textwidth}@{}}
\toprule
Analysis & Bank or selection rule & Evaluation and reuse \\
\midrule
Qwen3-0.6B lattices & Witness moves use earlier composing pairs; stratified
moves use metadata. The lattices share five moves. & Four-fold and small-support
prediction on fixed units; both fixed-prediction and full-refit bootstraps
use these historical units. \\
\addlinespace
Llama-3.2-3B lattice & Eight moves chosen from architecture metadata; 37
base/singleton/pair anchors. & 219 unqueried combinations on historical FIT
units. Every bootstrap draw recomputes anchors and predictors. \\
\addlinespace
Llama-3.2-1B one-shot & 224 alternatives and the primary selection policy
fixed before outer readout. & Original VALIDATION/TEST outcomes unread for
this Llama construction; the ARC development set was used in earlier Qwen studies. \\
\addlinespace
1B support-matched controls & Cap two is primary; caps four and eight are
sensitivity checks, fixed before new combination readouts. & Already exposed
VALIDATION/TEST units; the endpoint-beam state is held fixed. \\
\addlinespace
Qwen3-4B repair & Constructor rules fixed after the lattice audit; final state
selected on VALIDATION. & FIT/VALIDATION/TEST repartition previously evaluated
NLL blocks and ARC items. \\
\addlinespace
3B support-capped selection & Cap-two beam and additive control; caps four
and eight are fixed sensitivities. All states selected on FIT. & New NLL
TEST: 256 blocks from 123 articles, excluding the recorded 3B calibration,
FIT and later NLL-evaluation article sets. \\
\bottomrule
\end{tabular}
\end{table}

\section{Finite-Move Banks and Isolated Pricing}
\label{app:bank}

\subsection{The 97-move Qwen3 bank}

The bank contains 22 partial transitions in a 0.6B context $W_A$, 23 in $W_C$,
and 26 dose transitions along each of two consecutive solver steps on the
Qwen3-4B master. The level-refit step releases the levels
of a grid-legal hard state while keeping its codes fixed. The code re-selection
step exhaustively searches magnitude-ranked assignments
with free levels. A dose transition is a proxy-ranked batch of groups from
the full proposal; each is a legal same-rate move. Predictor and exact passes
ran as separate jobs, with endpoint outcomes unknown at prediction time.

\begin{table}[htb]
\centering
\caption{Sign and rank agreement with the exact functional difference on the 97
frozen moves; the midpoint response here is a central difference at the
midpoint (two forwards).}
\label{tab:predictor}
\vspace{0.4em}
\small
\begin{tabular}{@{}lccc@{}}
\toprule
Predictor & Sign & Pearson & Spearman \\
\midrule
Reconstruction proxy $d^\top\Lambda(m{-}W)$ & 54/97 & 0.269 & 0.089 \\
Current state $d^\top\nabla F(q_0)$ & 68/97 & 0.830 & 0.652 \\
Own midpoint $d^\top\nabla F(q_0+\tfrac12 d)$ & \textbf{96/97} & \textbf{0.9995} & \textbf{0.9994} \\
\bottomrule
\end{tabular}
\end{table}

\begin{figure}[htb]
\centering
\includegraphics[width=\textwidth]{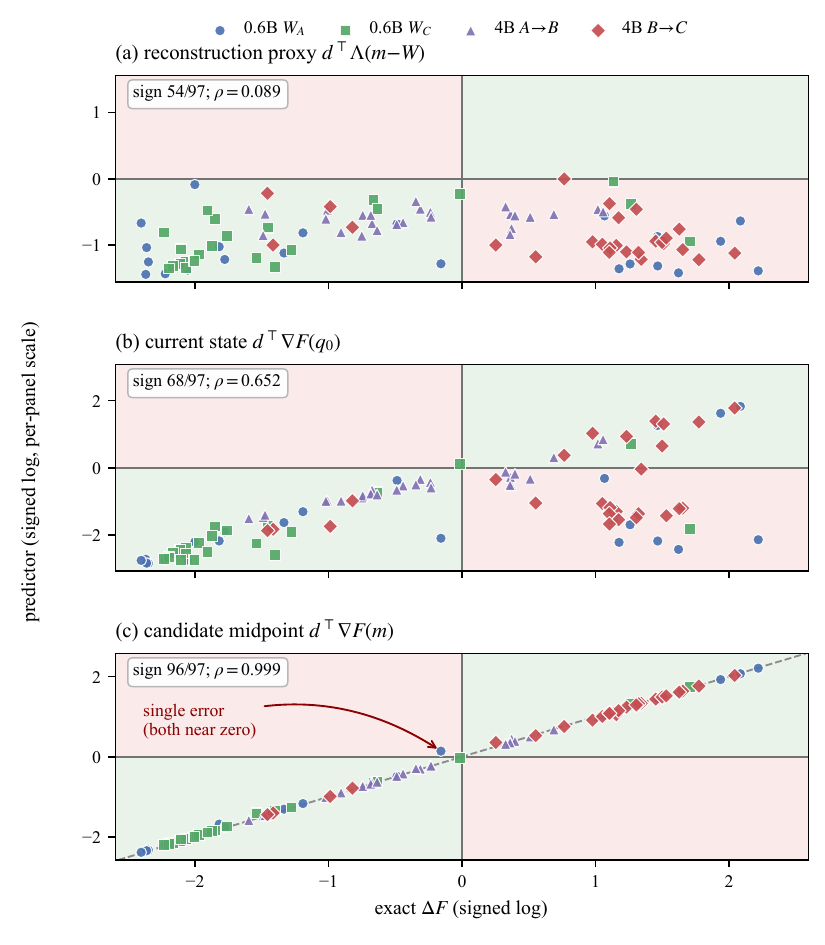}
\caption{Three predictors against the exact difference on the same 97 moves;
green quadrants are sign agreement.  (a) Nearly every proxy value is negative
because a solver only proposes moves its objective likes, which leaves the proxy
near chance at ordering them.  (b) The current-state gradient correlates but
misorients the 4B re-selection cluster.  (c) The midpoint response tracks the
exact difference across the full dynamic range.  Signed-log axes, per-panel $y$
transform; only sign and rank are compared.}
\label{fig:scatter}
\end{figure}

\begin{table}[htb]
\centering
\caption{Per-source results against the exact held-out functional difference.}
\label{tab:persource}
\vspace{0.4em}
\small
\begin{tabular}{@{}lrrrrrr@{}}
\toprule
Source & Moves & Midpoint & Current & Proxy & Pearson & Spearman \\
\midrule
0.6B $W_A$ & 22 & 21/22 & 16/22 & 13/22 & 0.99955 & 0.99774 \\
0.6B $W_C$ & 23 & 23/23 & 20/23 & 20/23 & 0.99912 & 0.99901 \\
4B level refit & 26 & 26/26 & 20/26 & 17/26 & 0.99994 & 0.99795 \\
4B code re-selection & 26 & 26/26 & 12/26 & 4/26 & 0.99967 & 0.99385 \\
\midrule
Total & 97 & 96/97 & 68/97 & 54/97 & 0.99953 & 0.99938 \\
\bottomrule
\end{tabular}
\end{table}

The midpoint response was estimated by a central difference at the midpoint;
the two step sizes ($h=1/8$ and $1/16$) agree in sign on all 97 moves.  Edge
sizes follow cumulative power-of-two boundaries plus the exact final eligible
count of each source ($W_A$ up to $524{,}288$ groups, $W_C$ $1{,}048{,}576$,
level refit $11{,}608{,}025$, re-selection $8{,}388{,}608$).  The single sign
error is $W_A$ move
\texttt{e10}: 512 groups, proxy $-1.44\times10^{-7}$, midpoint
$+4.58\times10^{-5}$, exact $-5.12\times10^{-5}$.
Median residuals $\abs{\text{midpoint}-\text{exact}}/\abs{\text{exact}}$ for the
four sources are $1.44\%$, $1.47\%$, $0.36\%$, $1.18\%$. Scoring all
97 moves took $9{,}683$ s on one GPU. The midpoint estimate stabilized on
fewer calibration blocks than direct endpoint differencing on 4 moves, the
same number on 90, and more on 3.

\subsection{The standard-GPTQ bank on Llama-3.2-1B}
We quantize \texttt{unsloth/Llama-3.2-1B} (16 layers, 112 projections) to
asymmetric INT3 with row-wise groups of 128, block size 128, dynamic per-group
min/max scale and zero point, and upper-Cholesky conditional error feedback.
Calibration is sequential: each layer receives inputs from the already
quantized preceding layers. Both models use 128 calibration sequences of
2048 WikiText-103 tokens and FP32 computation, with no activation ordering
or MSE clipping. Damping is $1\%$ for the GPTQ model and $0.1\%$ for the
alternative. Our kernel reproduces the reference
implementation~\citep{frantar2023gptq} to zero maximum decoded-weight
difference on a real $2048\times2048$ projection; generation takes about $166$ s
per model. Every compared state uses 3-bit codes and the same scale/zero
storage format, or $3.375$ bits per quantized projection weight. On the FIT
units, RTN $\to$ GPTQ($1\%$): NLL $3.519\to2.946$ ($-0.573$
$[-0.743,-0.350]$), ARC-Easy option-KL $0.259\to0.229$ ($-0.030$
$[-0.049,-0.011]$); GPTQ($1\%$)$\to$($0.1\%$): NLL $+0.439$ $[+0.103,+0.807]$,
ARC $-0.021$ $[-0.038,-0.005]$.

The audit bank fixes 98 legal interventions from these states: 75
RTN$\to$GPTQ($1\%$) and 23 GPTQ($1\%$)$\to$($0.1\%$), grouped by layer,
projection family and dose chain. Each is scored in its own left context.
NLL predictors use 28 FIT blocks; endpoints use those blocks and 28 held-out
WikiText-103 test blocks. ARC uses 592 label-free ARC-Easy development items.
The GPTQ($1\%$) calibration covariance supplies two reconstruction scores
per move: group-local and full contextual layer reconstruction, including
cross-group terms (Appendix~\ref{app:reconstruction}). Table~\ref{tab:crossarch}
compares the predictors. Same-unit NLL mean absolute error is $0.0338$
(current), $0.0431$ (Gauss--Newton), $0.00111$ (midpoint); ARC $0.0085$,
$0.0097$, $0.00086$. Held-out NLL versus ARC exact directions: both improve
37, NLL improves and ARC worsens 34, NLL worsens and ARC improves 12, both
worsen 15; on the 62 moves resolved on both, 23/25/6/8.  By source: RTN$\to$GPTQ
75 moves, midpoint 75/75 NLL (FIT), 74/75 ARC, 67/75 held-out NLL, conflicts
37/75 (27/50 resolved); GPTQ($1\%$)$\to$($0.1\%$) 23 moves, 22/23, 22/23,
19/23, conflicts 9/23.  NLL scoring $7{,}130$ s on one GPU.

\subsection{One-shot candidates and their original evaluation}
Each of the 112 projections takes GPTQ($1\%$), RTN or GPTQ($0.1\%$) codes
with the matching scale/zero. This gives 224 single replacements from the
GPTQ($1\%$) model, chosen without functional scores. FIT uses 28
training-split blocks and 592 items; VALIDATION uses 124 blocks and 297 items;
TEST uses 114 blocks and 297 items. These blocks are built from
\texttt{wikitext-103-raw-v1} with the Llama tokenizer: nonempty rows are
encoded without special tokens, joined by an encoded double newline in
source order, and cut into non-overlapping 2,048-token input blocks.
FIT follows the 128 producer-calibration blocks in the training stream.
VALIDATION retains all 124 complete validation-stream blocks; TEST retains
blocks 28--141, after the 28 test blocks already used in the isolated-move
audit. Incomplete final tails are dropped and exact duplicate blocks are
excluded. These are disjoint token blocks, without an article-disjoint claim.
Source indices and block hashes identify the fixed inputs.

The ARC items are the remaining 594 clean ARC-Easy development ids, assigned
alternately to VALIDATION and TEST in the frozen salted order. Their outcomes
had not been read for this Llama construction, although the development set
had prior Qwen exposure. The NLL outer blocks were also unread at this
original readout; later matched controls reuse both outer partitions.

The selector minimizes predicted $\Delta$ARC option-KL subject to predicted
$\Delta$NLL$\le0$, with one alternative per projection and no support limit,
solved to zero gap. Prices come from each replacement's own-midpoint
derivative, one current-state gradient row, or each exact FIT endpoint.
Candidates and the midpoint primary were frozen before VALIDATION/TEST.
The selectors choose 61/86/73 replacements. The midpoint state's changes are
VALIDATION $-0.0872$ $[-0.1340,-0.0409]$
NLL, $+0.0003$ $[-0.0287,+0.0288]$ ARC; TEST $-0.1172$ $[-0.2230,-0.0200]$
(perplexity $24.67\to21.94$, relative change CI $[-20.0\%,-2.0\%]$), $+0.0053$
$[-0.0193,+0.0295]$.  Current state: TEST $+0.4505$ $[+0.2621,+0.6589]$,
$+0.0108$ $[-0.0139,+0.0349]$.  Exact-singleton state: TEST $-0.1085$
$[-0.2218,-0.0028]$, $+0.0439$ $[+0.0169,+0.0714]$.  Paired midpoint minus
exact-singleton: TEST ARC $-0.0386$ $[-0.0535,-0.0236]$, NLL $-0.0086$
$[-0.0546,+0.0390]$; midpoint minus current: TEST NLL $-0.5677$
$[-0.7298,-0.4206]$.  Predicted composition on FIT from base $F_0=0.2292$:
$\Delta$ARC $-0.782$ (current), $-0.452$ (midpoint), $-0.473$ (exact
singleton), each below $-F_0$; KL non-negativity therefore bounds the exact
singleton state's FIT residual $R(S)=F(q_S)-F(q_0)-\sum_{g\in S}[F(q_g)-F(q_0)]$
by $R(S)\ge0.2437$ (measured below: $+0.524$).  Cost on FIT: one current row $43$ s; 224 own-midpoint
derivatives $5{,}274$ s; 224 exact endpoints $3{,}352$ s.

\subsection{Earlier fixed-scale solver variant}
A first bank on the same model used a frozen-scale code selection whose column
error feedback took $A^{-1}_{ij}/A^{-1}_{ii}$ from the full damped inverse
rather than GPTQ's Cholesky factor (the two differ at $3/48$ codes on a
controlled example), with the same 98-move design.  Held-out NLL signs: proxy
$59$, current $71$, Gauss--Newton $87$, midpoint $89$ (resolved 76: $47/59/70/74$);
ARC-Easy option-KL $43/88/86/95$ (resolved 70: $32/68/66/70$); conflicts $54/98$
and $28/52$ resolved.  Whole models on held-out blocks: RTN $3.558\to$ selection
$3.259$ ($-0.299$ $[-0.335,-0.262]$; ARC $+0.006$ $[-0.016,+0.028]$); the
$0.1\%$-damping variant $3.816$ ($+0.557$ $[+0.513,+0.597]$; ARC $+0.334$
$[+0.290,+0.379]$).  Runs of $7{,}874$ s and $16{,}826$ s on one RTX 5090.

\section{Exhaustive Lattices}
\label{app:atlas}

Each lattice has 8 primitives and 256 exact hard states, evaluated under
three functionals and probed at 14 centers. We use deterministic four-fold
subset cross-validation and order states by integer mask. Witness primitives
are bank batches $\{4,5,6,10,14,15,17,20\}$; stratified primitives
$\{2,3,6,10,14,17,20,22\}$ are chosen from group count, separation, edge
identity and a fixed random salt. Peak allocated memory is $26.7$ GB for D0,
$13.0$ GB for ARC and $18.5$ GB for HellaSwag.

\begin{table}[htb]
\centering
\caption{Stratified lattice, four-fold subset cross-validation: sign accuracy /
RMSE of each set-function model.}
\label{tab:atlascv}
\vspace{0.4em}
\small
\begin{tabular}{@{}lccc@{}}
\toprule
Model & D0 NLL & ARC KL & Hella KL \\
\midrule
additive & 100.0\% / 0.00292 & 90.6\% / 0.0766 & 88.3\% / 0.0505 \\
product coverage & 71.1\% / 0.0161 & 55.9\% / 0.1195 & 65.2\% / 0.0684 \\
full quadratic & 100.0\% / 0.00052 & 92.6\% / 0.0301 & 93.4\% / 0.0153 \\
sparse $k{=}8$ & 100.0\% / 0.00071 & 94.1\% / 0.0309 & 93.4\% / 0.0180 \\
sparse $k{=}16$ & 100.0\% / 0.00051 & 94.5\% / 0.0296 & 93.4\% / 0.0154 \\
low-rank $r{=}1$ & 100.0\% / 0.00151 & 93.8\% / 0.0559 & 91.4\% / 0.0307 \\
low-rank $r{=}4$ & 100.0\% / 0.00052 & 92.6\% / 0.0314 & 94.1\% / 0.0157 \\
\bottomrule
\end{tabular}
\end{table}

\begin{table}[htb]
\centering
\caption{Pareto-front recall / precision.  Exact fronts: 11 states (witness),
30 states (stratified).}
\label{tab:pareto}
\vspace{0.4em}
\small
\begin{tabular}{@{}lcc@{}}
\toprule
Model & witness & stratified \\
\midrule
additive & 54.5\% / 85.7\% & 36.7\% / 91.7\% \\
product coverage & 0\% / 0\% & 10.0\% / 75.0\% \\
full quadratic & 63.6\% / 100\% & 86.7\% / 86.7\% \\
sparse $k{=}8$ & 72.7\% / 100\% & 86.7\% / 78.8\% \\
sparse $k{=}16$ & 63.6\% / 100\% & 83.3\% / 89.3\% \\
low-rank $r{=}1$ & 63.6\% / 100\% & 56.7\% / 73.9\% \\
low-rank $r{=}4$ & 72.7\% / 100\% & 90.0\% / 87.1\% \\
\bottomrule
\end{tabular}
\end{table}

Exact common-descent states: 69/256 (witness), 48/256 (stratified).  Under a
raw max--min rule the stratified lattice's additive, quadratic, sparse, low-rank
and recentered constructors all select mask 159, batches $\{2,3,6,10,14,22\}$,
with exact deltas $(-0.0268,-0.1698,-0.0697)$; on the witness lattice the additive
model selects mask 63 $(-0.0404,-0.4469,-0.1299)$ and the quadratic family mask
182 $(-0.0756,-0.1384,-0.0773)$, both on the front.  Sparse $k{=}4$ produced one
false-safe state on the witness lattice (predicted three negatives, exact ARC KL
$+0.0026$).

Center-probe tomography with four random centers and rank-4
truncation: sign accuracy $78.5/84.4/90.2\%$ (witness) and $79.3/94.1/93.0\%$
(stratified); with all thirteen non-zero probes the stratified ARC and HellaSwag
RMSE are $0.0775$ and $0.0489$ against $0.0301$ and $0.0153$ for the
endpoint-fitted quadratic.

In the cross-functional pair coordinates ($3\times28$), the first two
singular directions carry $99.89\%$ (witness) and $99.90\%$ (stratified) of
the energy. Since this matrix has rank at most three, we examine spectral
concentration within each functional using the fitted matrix
$\widehat\beta$ in Table~\ref{tab:atlas}. Dividing $\widehat\beta_{ik}$ by the
decoded-weight norms $\|d_i\|_2\|d_k\|_2$ of the two moves leaves the
participation ranks at $2.04$--$2.64$ with the top four spectral components
carrying at least $99.67\%$ of the energy in every cell (normalizing by
$\sqrt{G_iG_k}$ group counts: $2.04$--$2.32$); on the stratified ARC cell,
however, $97.8\%$ of the normalized squared energy is incident to one two-group
move. The star pattern alone concentrates the spectrum in this cell; it does
not establish a shared latent subspace across moves.

The product-coverage baseline in Tables~\ref{tab:atlascv} and~\ref{tab:pareto}
fits $c[1-\prod_i(1-a_i x_i)]$, with $0\le a_i<1$ and a free signed scale $c$.
Its projected-gradient fit uses four constant-rate initializations
($0.05,0.2,0.5,0.8$). The monotone single-index control below is a different,
more flexible model: it learns a signed direction $a$ and an isotonic link
$h$ in $h(a^\top x)$. Both are evaluated by held-out subset predictions;
their different front recalls reflect different model and fitting constraints.

\begin{table}[htb]
\centering
\caption{Set-function controls on the two lattices.  Left: four-fold subset
cross-validation of a multi-start monotone single-index fit $h(a^\top x)$
(isotonic link, Isotron direction updates, starts at the additive direction, the
coordinate axes and 32 frozen random directions) against the quadratic family.  Right: Hamming-stratified extrapolation on the stratified
lattice, fitted on the 93 states of support $\le3$ and evaluated on the 163
states of support $4$--$8$.  FS: false-safe states (predicted improving on all
three functionals, exact worsening on at least one).}
\label{tab:controls}
\vspace{0.4em}
\footnotesize
\setlength{\tabcolsep}{3.5pt}
\begin{tabular}{@{}lccc@{}}
\toprule
& \multicolumn{2}{c}{4-fold CV Pareto recall (FS)} & Support $\le3\to4$--$8$, stratified \\
\cmidrule(lr){2-3}\cmidrule(lr){4-4}
Model & witness & stratified & Pareto recall / precision (FS) \\
\midrule
additive (zero intercept) & 54.5\% & 36.7\% & 52.0\% / 86.7\% (1) \\
affine (free intercept) & 63.6\% & 30.0\% & 52.0\% / --- \\
monotone single-index & 45.5\% (8) & 30.0\% (11) & 32.0\% / 80.0\% (1) \\
full quadratic & 63.6\% & 86.7\% & \textbf{92.0\% / 100\%} (0) \\
sparse $k{=}16$ & 63.6\% & 83.3\% & 92.0\% / 88.5\% (0) \\
low-rank $r{=}4$ & 72.7\% & 90.0\% & 80.0\% / 87.0\% (0) \\
\bottomrule
\end{tabular}
\end{table}

On the extrapolation split the full quadratic lowers the ARC RMSE from $0.0997$
to $0.0470$ and the HellaSwag RMSE from $0.0719$ to $0.0245$ relative to the
additive model.  Free-intercept affine control (same folds and splits): witness
pooled recall $63.6\%$ against $54.5\%$ zero-intercept and $63.6\%$ quadratic;
stratified $30.0\%$ against $36.7\%$ and $86.7\%$; low-to-high extrapolation
$80.0\%/52.0\%$ (witness/stratified) for both additive variants.  Full-fit
recall on the witness lattice is $72.7/81.8/63.6\%$ (zero-intercept / affine /
quadratic), so the witness front does not favour the quadratic under every
metric.

In the reverse direction (fit on support $\ge5$, evaluate on support $\le3$)
the quadratic family recovers $100\%$ of the witness front against $50\%$ for
the additive model and $67$--$100\%$ of the stratified front against $33\%$,
with one to four false-safe states for some quadratic variants. On the witness
lattice the low-to-high direction gives $80\%$ recall for both the additive
and quadratic models, with false-safe states $3\to0$.

\begin{table}[htb]
\centering
\caption{Two exhaustive $256$-state lattices, three functionals: Walsh energy by
degree (share of the non-constant energy), $R^2$ of the full quadratic fit,
participation rank of its $8\times8$ fitted pairwise matrix $\widehat\beta$, and four-fold
cross-validated recall of the exact Pareto front.  The stratified lattice was
selected without any functional outcome.}
\label{tab:atlas}
\vspace{0.3em}
\footnotesize
\setlength{\tabcolsep}{4pt}
\begin{tabular}{@{}llrrrrrrr@{}}
\toprule
& & \multicolumn{2}{c}{Walsh energy} & & & \multicolumn{3}{c}{Pareto recall} \\
\cmidrule(lr){3-4}\cmidrule(lr){7-9}
Lattice & Functional & deg.\ 2 & deg.\ $\ge3$ & $R^2$ quad. & p-rank $\widehat\beta$
  & additive & quadratic & low-rank 4 \\
\midrule
witness    & D0 NLL   & 1.70\% & 0.07\% & 0.9992 & 2.91 & \multirow{3}{*}{54.5\%} & \multirow{3}{*}{63.6\%} & \multirow{3}{*}{72.7\%} \\
witness    & ARC KL   & 6.67\% & 1.86\% & 0.9802 & 3.47 & & & \\
witness    & Hella KL & 11.43\% & 2.85\% & 0.9688 & 2.66 & & & \\
\midrule
stratified & D0 NLL   & 1.20\% & 0.03\% & 0.9997 & 2.24 & \multirow{3}{*}{36.7\%} & \multirow{3}{*}{86.7\%} & \multirow{3}{*}{\textbf{90.0\%}} \\
stratified & ARC KL   & 7.01\% & 0.92\% & 0.9907 & 3.42 & & & \\
stratified & Hella KL & 12.74\% & 1.10\% & 0.9889 & 2.45 & & & \\
\bottomrule
\end{tabular}
\end{table}

\paragraph{Möbius-anchored quadratic.}
The M\"obius-anchored quadratic takes the exact coefficients $\beta$ of
\eqref{eq:moebius} from the empty state, the 8
singleton and the 28 pair endpoints and predicts the remaining 219 states with
no fit.  On the witness lattice it recovers 10 of the 11 exact front states (15
predicted, precision $66.7\%$) with RMSE $0.00363/0.1635/0.0786$ (D0 / ARC /
HellaSwag) on those 219 states, against $0.00489/0.1922/0.1187$ for the
exact-singleton additive model and $0.00067/0.0463/0.0267$ for the
zero-intercept least-squares quadratic fitted to all 256 (in sample); on the
stratified lattice it recovers 27 of 30 (29 predicted, precision $93.1\%$) with
RMSE $0.00131/0.0507/0.0306$ against $0.00443/0.1238/0.0965$ and
$0.00044/0.0261/0.0140$.  Sign agreement on the 219 states: witness
$98.6/89.5/86.3\%$ (additive $95.4/89.0/84.0$, fitted $99.1/94.5/95.4$);
stratified $100/95.9/96.3\%$ (additive $100/90.9/84.5$, fitted
$100/93.6/95.4$). Higher-order M\"obius terms can project onto the linear and
pairwise coefficients of a least-squares fit, separating $\widehat\beta$ from
the anchored $\beta$. The Walsh degree energy in Table~\ref{tab:atlas} belongs
to an orthogonal decomposition; it does not bound the error of M\"obius
truncation. The unqueried-state errors above measure that truncation directly.

\paragraph{Contextual reversals.}
For each lattice and functional, all $8\times128=1{,}024$ same-edit contrasts
$f(S{\cup}\{i\})-f(S)$ and $28\times64=1{,}792$ ordering contrasts
$f(S{\cup}\{i\})-f(S{\cup}\{j\})$ were tested with $10{,}000$ paired bootstrap
draws over the shared statistical units (token-weighted blocks for NLL,
individual items for the option-KL functionals); the maximum standardized
contrast error within each cell sets the simultaneous critical value
($3.92$--$4.17$ at $\alpha=0.05/6$ across the six cells), and a witness requires
both contexts resolved with opposite signs.  Resolved same-edit sign reversals
out of 8 edits / ordering reversals out of 28 pairs: stratified D0 $0/1$, ARC
$6/14$, HellaSwag $4/11$; witness D0 $0/6$, ARC $5/14$, HellaSwag $4/12$.  Under
$F(S)=h(\sum_{i\in S}a_i)$ with monotone $h$, neither reversal can occur, since
the index increment of edit $i$ is $a_i$ in every context and the index
difference of $S{\cup}\{i\}$ against $S{\cup}\{j\}$ is $a_i-a_j$.

\paragraph{Front recall under resampling, quality tolerance and fixed choice.}
With pooled four-fold predictions held fixed, $2{,}000$ paired resamples of
the statistical units (Monte Carlo error at most $1.1$ points) give
mean stratified recall $35.9\%$ (additive), $33.3\%$ (affine), $84.4\%$ (full
quadratic) and $87.8\%$ (rank 4); witness $59.0/68.0/67.0/77.1\%$.

On the original sample, quality-tolerant coverage asks whether an exact endpoint on the predicted front
is within $0$, $1$ or $2$ median state standard errors of every exact-front
endpoint on each functional: stratified $66.7/83.3\%$ (additive),
$70.0/86.7\%$ (affine) and $100/100\%$ for both quadratics at one and two
standard errors; on the witness lattice every model reaches $100\%$ at one.
These tolerances describe proximity in loss space, not deployment equivalence.

To select a single state, we minimize the worst standardized functional among
those predicted to improve all three, returning state $0$ if none qualify.
Standardization uses the lattice's exact-state standard deviations, fixed
across models and resamples. We measure regret against the exact best feasible
state. On the stratified lattice, additive and affine select mask 150
(regret $0.0716$; on the exact front in
$24.8\%$ of resamples), full quadratic mask 21 ($0.0396$; $99.95\%$), rank 4 mask
149 ($0.0855$; $43.3\%$); all three selected states improve all three
functionals in every resample.  Witness: additive, affine and full quadratic
select mask 55 ($0.0267$), rank 4 mask 54 ($0.0184$).

\paragraph{Refitting and re-selection under resampling.}
\label{app:qwen_refit}
We repeat the analysis with model fitting and state selection inside each
resample. The two lattices retain their original four-fold assignments.
Within each functional, a draw resamples the same units across all 256 states
and both predictors: 64 token blocks for NLL, 592 individual ARC items and
400 individual HellaSwag items. NLL retains its predicted-token weights.
We use 2,000 draws with seed 20260919. Each draw recomputes the three loss
coordinates, refits the zero-intercept additive and full quadratic models
in the four folds, pools their out-of-fold predictions, and rebuilds the
exact and predicted fronts. It then repeats the standardized worst-functional
selection rule above. The exact feasible oracle is recomputed on the same draw.
A matched control evaluates the original fixed predictions on these same draws.

Both the utility scales and quality tolerances remain fixed across draws.
In NLL/ARC/HellaSwag order, the utility-scale vectors are
$(0.02530861,0.26377074,0.12957188)$ for the stratified lattice and
$(0.02590839,0.32633578,0.14895942)$ for the witness lattice.
The corresponding one-SE tolerances are
$(0.00265320,0.01188189,0.00839176)$ and
$(0.00279485,0.01333041,0.01101104)$.
A true-front state is covered when an actual endpoint on the predicted front
is no worse by more than this tolerance in every functional.
Single-state regret instead evaluates the exact standardized utility of the
state selected from the predictions. These resamples measure stability on the
original evaluation units.

\begin{table}[!htb]
\centering
\caption{Qwen lattice stability with refitting and re-selection.
Entries are bootstrap means; intervals are paired quadratic-minus-additive
percentile intervals over the same 2,000 draws. Coverage differences are
percentage points. Regret uses the fixed standardized utility.}
\label{tab:qwen_refit}
\small
\setlength{\tabcolsep}{5pt}
\begin{tabular}{@{}lrrl@{}}
\toprule
Lattice & Additive & Full quadratic & Paired difference, $95\%$ interval \\
\midrule
\multicolumn{4}{@{}l}{One-SE actual-endpoint quality coverage (\%)} \\
Stratified & 70.06 & 99.91 & $[+20.69,+40.00]$ \\
Witness & 96.14 & 91.01 & $[-22.22,0.00]$ \\
\addlinespace
\multicolumn{4}{@{}l}{Selected-state utility regret} \\
Stratified & 0.08840 & 0.03382 & $[-0.15707,+0.01792]$ \\
Witness & 0.06460 & 0.03800 & $[-0.16842,0.00000]$ \\
\bottomrule
\end{tabular}
\end{table}

The stratified lattice's quality-coverage advantage survives refitting,
while the paired interval for the single-state regret difference spans zero.
The original quadratic choice, mask 21, is selected again in $38.0\%$ of draws.
On the witness lattice, the additive front has higher mean one-SE coverage.

\paragraph{Response operator against exact pairwise coefficients.}
We compare exact coefficients $\beta_{ik}=f(e_i{+}e_k)-f(e_i)-f(e_k)$ with
the response operator fitted by $2[r(z)-r(0)]=Bz$ on the thirteen non-zero
centers. The diagonal is free and absorbed into the linear term; comparisons
use off-diagonal entries. Sign agreement is
$92.9/64.3/78.6\%$ (witness D0/ARC/Hella) and $67.9/82.1/85.7\%$ (stratified),
Spearman $0.67$--$0.94$, normalized RMSE $0.51$--$0.81$; $3.6$--$32.1\%$ of the
pairwise discrepancies fall within $1.96$ paired-bootstrap standard errors
($1{,}000$ replicates, joint resampling of blocks and microbatches).  The
operator recovered from the eight unit centers alone mispredicts the response at
the full and four random centers on $31/39/34$ (witness) and $34/39/37$
(stratified) of 40 components with paired $95\%$ radii excluding zero. The
residual grows with the center's distance from the base state. The
gauge-aligned response quadratic nonetheless reproduces $97.7/92.6/92.2\%$ and
$100/94.1/94.5\%$ of exact endpoint signs.

\paragraph{Common-center error against the interaction cut.}
\label{app:cut_decomposition}
We compare the measured common-center error with the half-cut directly on
both Qwen3 lattices. Let $s$ and $\beta$ be the exact M\"obius coefficients of
Equation~\eqref{eq:moebius}, and define
$\eta_i=v_i-s_i-\tfrac12\sum_{k\ne i}\beta_{ik}$ using the measured
full-center response $v$. Without a constant-curvature assumption,
\begin{equation}
 v^\top x-f(x)
   =\tfrac12x^\top\beta(\onev-x)+\eta^\top x-R_{\ge3}(x).
 \label{eq:general_cut}
\end{equation}
Substituting $v=s+\tfrac12\beta\onev+\eta$ and
$f(x)=s^\top x+\tfrac12x^\top\beta x+R_{\ge3}(x)$ gives the identity.
The empirical question is whether the cut alone approximates the observed
error. All endpoint and response units are paired: 64 NLL blocks,
74 ARC microbatches and 50 HellaSwag microbatches. We retain token and item
weights and use the exact endpoint pair coefficients, rather than a fitted
quadratic or response-derived operator.

\begin{table}[!htb]
\centering
\caption{Common-center error decomposition on the same units.
$E=v^\top x-f(x)$, $C=\tfrac12x^\top\beta(\onev-x)$.
Each RMS is over all 255 non-baseline states, in nats; the last column is
$\lVert E-C\rVert_2/\lVert E\rVert_2$. Components can cancel, so their RMS
values are not additive shares of error.}
\label{tab:common_cut}
\footnotesize
\setlength{\tabcolsep}{4pt}
\begin{tabular}{@{}llrrrrr@{}}
\toprule
Lattice & Functional & RMS $E$ & RMS $C$ & RMS $\eta^\top x$ & RMS $R_{\ge3}$ & Relative RMSE \\
\midrule
Witness & NLL & 0.00686 & 0.00371 & 0.00956 & 0.00337 & 1.109 \\
& ARC & 0.15921 & 0.19924 & 0.39665 & 0.15153 & 1.916 \\
& HellaSwag & 0.08579 & 0.10718 & 0.19837 & 0.07283 & 1.863 \\
Stratified & NLL & 0.00328 & 0.00359 & 0.00370 & 0.00121 & 1.019 \\
& ARC & 0.10686 & 0.06813 & 0.10644 & 0.04695 & 0.991 \\
& HellaSwag & 0.05448 & 0.06240 & 0.06773 & 0.02836 & 1.140 \\
\bottomrule
\end{tabular}
\end{table}

The half-cut alone leaves at least $99\%$ of the original RMS error in
these cells. The endpoint remainder and the common-center allocation
discrepancy must both be retained. The latter is defined by $\eta$; it does
not uniquely isolate curvature drift. The identity closes to numerical
precision on these measurements; the empirical result is the size of its terms.

\paragraph{Anchored low-rank search and its anchor-budget simulation.}
The anchored variant of the beam takes the backbone of $B$ from antithetic probe
pairs $z^\pm=(\onev\pm u)/2$, for which $2[r(z^+)-r(z^-)]=Bu$, fits the
functional-specific residual on $4$--$12$ exact pair endpoints chosen before
their values are read, and solves the quadratic pseudo-Boolean problem with the
same beam; it recovers $83\%$ of the stratified and $55\%$ of the witness front.
On the stratified lattice, we evaluate 315 fixed configurations: exact pair
budgets 4, 8, 12; anchor selection by response disagreement, leverage,
manifest diversity, a hybrid or a hash control; unit-, all- or random-center
rank-4 backbones; and affine, subspace or backbone-plus-exact-residual models.
Anchors are chosen before any unqueried pair value is read.

The quality configuration uses 12 disagreement anchors, the all-center
backbone and exact residuals. It achieves Pareto recall
$83.3\%$, precision $89.3\%$, 0 false-safe, 1 false-reject, emitted mask 156
$=\{6,10,14,22\}$ with exact deltas $(-0.0256,-0.1927,-0.0843)$ and zero regret
under the response-standardized utility; additive response-only search: recall
$43.3\%$, emitted $\{2,3,6,14,22\}$ with regret $0.080$.  Eight leverage
anchors: recall $90.0\%$, precision $84.4\%$; four diversity anchors: $80.0\%$
recall, no improvement of the emitted state. The same quality configuration
on the witness lattice gives recall $54.5\%$, 4 false-safe and 8 false-reject.
Forward-equivalent cost
per functional including one confirmation: $106$--$155$ against $37$ for exact
pair enumeration on eight primitives.  An exact-endpoint oracle beam that reads
endpoints during search finds $63\%$, $97\%$ and $100\%$ of the front at widths
4, 8 and 12 using 69, 95 and 104 unique endpoints; it is the active-query
reference for any recentered method.

\paragraph{Recentered search backend (Llama-3.2-1B, INT3).}
A categorical state over 16 disjoint proxy-dose batches with three code
alternatives each (RTN, the code selection at $1\%$ and at $0.1\%$ damping) supports add, remove
and one-for-one replacement, exact materialization of any mixed state, and hooked
backward rows that agree with the scalar current-state and own-midpoint
derivatives to $10^{-9}$--$10^{-11}$.  On two FIT blocks the cheap screen
(current row, proxy, structural diversity, deterministic random) ranks the exact
best of 32 atomic moves first, and screening plus a top-4 exact judge costs
$126$ s against $410$ s for enumerating all 32 endpoints (break-even at $K=26$);
a per-move Gauss--Newton screen signs $31/32$ moves correctly yet ranks the exact
best $29$th.  A closed-form legal scale refit on a 1024-group batch lowers the
weight squared error by $9.06\%$ and raises FIT NLL by $0.0018$, so refit enters
the search as a candidate judged at exact endpoints rather than as a default.
Four own-midpoint anchor rows drift from the base row by $8.7\times10^{-4}$ on
average with 2 of 128 sign changes.

\paragraph{Manifest structural kernel.}
Each primitive is described by group count, separation, occupancy ratio,
strong Hamming, layer histogram, projection-family histogram, layer centroid
and spread, proxy scale, and own-singleton response. We test nine-feature
ridge and symmetric bilinear models of rank one and two. The two lattices share five of
eight primitives, so transfer is scored on the eighteen pair identities unseen
in training.  Witness$\to$stratified best sign agreement $61.1/50.0/44.4\%$
(D0/ARC/Hella), stratified$\to$witness $77.8/55.6/66.7\%$, with Spearman between
$-0.19$ and $0.70$ and no model best in both directions on all three
functionals; substituting the
predicted $\beta$ into the quadratic gives Pareto recall $23$--$33\%$ or four to
ten false-safe states.

\subsection{Metadata-selected standard-GPTQ lattice on Llama-3.2-3B}
\label{app:independent_l8}

The base is Llama-3.2-3B quantized by standard GPTQ with $1\%$ damping,
activation ordering and static groups. Alternatives come from RTN and
GPTQ with $0.1\%$ damping. The producer, original-order g128 grouping and
INT3 encoding are those of the 3B audit in Appendix~\ref{app:protocols}.
Each move replaces a complete projection, including its codes, scales and
zero points, preserving the $3.375$-bit decoder payload per quantized weight.
Eight moves are selected using architecture metadata and a fixed random salt:
layer 5 q-projection to RTN and gate-projection to GPTQ($0.1\%$);
layer 8 k-projection to GPTQ($0.1\%$) and up-projection to RTN;
layer 15 v- and down-projections to RTN and GPTQ($0.1\%$), respectively;
layer 24 o- and gate-projections to GPTQ($0.1\%$) and RTN, respectively.
Selection uses no functional or reconstruction scores.

All $256$ states are evaluated on $28$ WikiText-103 FIT blocks and $592$
ARC-Easy FIT items. Option-KL uses the external floating-point Llama model as
its reference and the full-continuation scoring of
Equation~\eqref{eq:option_kl_definition}. The base losses are $2.280813741$
NLL and $0.369477886$ option-KL. The combinations are newly evaluated, while
the statistical units and singleton bank were used in the earlier 3B audit.
This experiment tests transfer of the combination structure to a different
bank, with prediction and exact evaluation on the same statistical units.

\paragraph{Anchors and unqueried combinations.}
The M\"obius quadratic takes $s_i$ and $\beta_{ik}$ directly from the base,
eight singleton and 28 pair endpoints in Equation~\eqref{eq:moebius}.
These $37$ queries determine its predictions without a fit. Table~\ref{tab:independent_l8}
evaluates the remaining $219$ states with support at least three, against
the additive prediction using the same exact singleton effects.
On this evaluation set, the quadratic recovers all six Pareto states with
$100\%$ precision and recall; addition recovers five of six, also with six
predicted states. Of the five exact jointly improving states, both models
recover all five, but addition falsely predicts joint improvement for a sixth
(mask 137). The quadratic makes neither a false-safe nor a false-reject error.
It also finds the exact minimizer of each functional within the evaluation set;
addition misses the ARC minimizer by $0.001336$.

\begin{table}[!htb]
\centering
\caption{Predicting unqueried combinations on the Llama-3.2-3B
GPTQ lattice. The first block uses base/singleton/pair anchors and tests the
219 states of support $3$--$8$. The second fits on the 93 states of support
$\le3$ and tests the 163 states of support $4$--$8$. All errors use exact
effects on the same statistical units; sign counts concern changes from the
base. Lower errors and higher sign agreement/Spearman are better.}
\label{tab:independent_l8}
\small
\setlength{\tabcolsep}{4pt}
\begin{tabular}{@{}llrrrr@{}}
\toprule
Functional & Predictor & MAE & Max. error & Signs & Spearman \\
\midrule
\multicolumn{6}{@{}l}{\textit{Exact anchors; 219 unqueried combinations}} \\
NLL & singleton additive & $1.467\!\times\!10^{-4}$ & $4.268\!\times\!10^{-4}$ & 215/219 & 0.9924 \\
NLL & pair quadratic & $1.835\!\times\!10^{-5}$ & $7.536\!\times\!10^{-5}$ & 219/219 & 0.9997 \\
ARC KL & singleton additive & $4.110\!\times\!10^{-3}$ & $1.565\!\times\!10^{-2}$ & 207/219 & 0.9446 \\
ARC KL & pair quadratic & $7.843\!\times\!10^{-4}$ & $4.274\!\times\!10^{-3}$ & 218/219 & 0.9975 \\
\midrule
\multicolumn{6}{@{}l}{\textit{Small-support fit; 163 larger combinations}} \\
NLL & additive & $1.146\!\times\!10^{-4}$ & $2.981\!\times\!10^{-4}$ & 161/163 & 0.9948 \\
NLL & quadratic & $1.605\!\times\!10^{-5}$ & $5.788\!\times\!10^{-5}$ & 163/163 & 0.9997 \\
ARC KL & additive & $3.039\!\times\!10^{-3}$ & $1.143\!\times\!10^{-2}$ & 158/163 & 0.9678 \\
ARC KL & quadratic & $6.222\!\times\!10^{-4}$ & $2.927\!\times\!10^{-3}$ & 163/163 & 0.9982 \\
\bottomrule
\end{tabular}
\end{table}

The front restricted to the $219$ unqueried states has masks
$\{13,37,69,76,133,165\}$; the additive model omits 76 and adds 141.
The all-state front is a different set,
$\{5,12,13,37,133,165\}$, including two queried pairs.
On that set the quadratic also has $100\%$ recall and precision, while the
additive model omits 12 and adds 141. Front membership and joint-improvement
counts here refer to exact sample means.

\paragraph{How pair information changes a decision.}
\label{app:l8_example}
Mask 137 replaces the layer-5 q-projection, layer-8 up-projection and layer-24
gate-projection by their RTN encodings. Exact singleton addition predicts that
both losses decrease (Table~\ref{tab:l8_example}). The ARC pair terms sum to
$+0.005503$, more than offsetting the predicted $-0.002242$ improvement.
The resulting prediction, $+0.003262$, tracks the actual $+0.003757$ increase.
Thus the difference between the decisions comes entirely from composition:
the two predictors use identical exact isolated effects.

\begin{table}[!htb]
\centering
\caption{A three-move combination whose predicted ARC direction changes after
accounting for pair interactions. All entries are changes from the base;
negative values lower the loss.}
\label{tab:l8_example}
\small
\setlength{\tabcolsep}{8pt}
\begin{tabular}{@{}lrr@{}}
\toprule
Quantity & $\Delta$ NLL & $\Delta$ ARC-KL \\
\midrule
Exact singleton sum & $-4.281\!\times\!10^{-6}$ & $-0.002242$ \\
Pair correction & $-3.919\!\times\!10^{-5}$ & $+0.005503$ \\
Singleton sum + pair correction & $-4.347\!\times\!10^{-5}$ & $+0.003262$ \\
Exact composed endpoint & $-5.386\!\times\!10^{-5}$ & $+0.003757$ \\
\bottomrule
\end{tabular}
\end{table}

\begin{samepage}
For this illustrative state, $10{,}000$ paired bootstrap resamples (seed 20260912) of the
28 NLL blocks and 592 ARC items give descriptive $95\%$ intervals
$[-0.001620,+0.001508]$ for exact NLL and $[-0.002879,+0.010537]$ for exact
ARC-KL. The ARC pair correction has interval $[+0.003697,+0.007364]$.
The example's exact endpoint directions are therefore unresolved by these
intervals, while the measured interaction shifts ARC upward.\par
\end{samepage}

\paragraph{Support extrapolation and interaction degree.}
A separate least-squares comparison fits the $93$ states with support at most
three and predicts the $163$ remaining states. Both additive and quadratic
models recover all six Pareto states \emph{within the latter set}, with
precision $85.7\%$ and $100\%$, respectively. Addition misses one of its two
jointly improving states; the quadratic misses neither and both make zero
false-safe predictions. The quadratic reduces mean absolute error by
$7.1\times$ for NLL and $4.9\times$ for ARC.
Non-constant Walsh energy is $99.3393/0.6510/0.0096\%$ for NLL and
$97.5392/2.4024/0.0584\%$ for ARC at degrees one, two, and three or higher.
For the full eight-move state, the residual relative to exact singleton
addition is $+0.00039755$ NLL and $+0.00923372$ ARC option-KL.
Pairwise terms account for most of the nonadditivity and improve prediction
despite their small share of total variation.

\paragraph{Contextual reversals and uncertainty.}
The same-edit and ordering audit uses the $2{,}816$ contrasts per functional
defined above, $10{,}000$ paired unit-bootstrap draws, a within-functional
max-$t$ critical value and Bonferroni allocation across the two functionals.
Point estimates show same-edit sign reversals for two of eight moves on
each functional and ordering reversals for two of 28 NLL pairs and ten ARC
pairs. None has opposite signs resolved simultaneously on both sides.
Thus the new bank reproduces composition error and low-order prediction;
the resolved contextual-reversal witnesses remain those of the Qwen3 lattices.
All states are evaluated without training. Enumeration takes $10{,}731$ s of
forward computation ($10{,}799$ s total), with $15.77$ GB peak allocated CUDA
memory on one RTX 5090.

\paragraph{Refitted-bootstrap front stability.}
\label{app:l8_refit}
We resample the historical FIT units 2,000 times (seed 20260912).
Within each functional, the same sampled indices apply to every one of the
256 states. NLL resamples blocks and retains their predicted-token weights;
ARC resamples items. The two datasets use independent random streams.
Each draw recomputes the base, singleton and pair means and reconstructs
the anchored quadratic, then compares its predictions with the resampled
exact endpoints. The least-squares models are also refitted from their
93 support-at-most-three states.

Table~\ref{tab:l8_refit} reports front membership over all 256 states,
including anchors. MAE and false-safe counts are evaluated only on the 219
unqueried states (163 for the least-squares extrapolation models).
This all-state front is distinct from the support-restricted front in
Table~\ref{tab:independent_l8}. A false-safe prediction has both predicted
effects negative and at least one exact effect positive in that draw.
The intervals below are bootstrap variability intervals conditional on the
observed data; they do not certify a population Pareto set.

\begin{table}[!htb]
\centering
\caption{Recomputing the measurements and predictors in every bootstrap draw.
Recall and precision concern the all-state sample-mean front; false-safe is
the mean count on unqueried states. Error columns are per-functional MAE.}
\label{tab:l8_refit}
\footnotesize
\setlength{\tabcolsep}{4pt}
\begin{tabular}{@{}lrrrrr@{}}
\toprule
Predictor & Recall (\%) & Precision (\%) & NLL MAE & ARC MAE & False-safe \\
\midrule
Exact-singleton additive & 81.93 & 73.11 & $1.752\times10^{-4}$ & $4.281\times10^{-3}$ & 1.5595 \\
37-anchor quadratic & 98.03 & 98.37 & $2.805\times10^{-5}$ & $8.518\times10^{-4}$ & 0.1080 \\
Support-$\le3$ additive fit & 88.67 & 77.15 & $1.345\times10^{-4}$ & $3.152\times10^{-3}$ & 0.4355 \\
Support-$\le3$ quadratic fit & 98.34 & 97.88 & $2.343\times10^{-5}$ & $6.729\times10^{-4}$ & 0.0545 \\
\bottomrule
\end{tabular}
\end{table}

The $95\%$ recall intervals are $[44.44,100.00]\%$ for exact-singleton
addition and $[80.00,100.00]\%$ for the anchored quadratic. Keeping the
original predictions fixed instead of recomputing them gives mean recalls
$64.28\%$ and $79.92\%$. These answer different questions: stability of a
fixed prediction versus stability of the measurement-and-prediction procedure.

\paragraph{Near-front quality and a fixed utility.}
For each functional $j$, let $\sigma_j$ be a pooled standard error of paired
state-minus-base effects, fixed from the original unit table. We pool the
weighted within-state variances across all states and divide by the
effective weighted unit count. A true Pareto point $p$ is covered at
tolerance $a$ if a state $q$ on the predicted front has actual effects
$f_j(q)\le f_j(p)+a\sigma_j$ for both functionals. We retain the ordinary
fronts: the tolerance changes the coverage relation, not dominance or which
states are returned.

\begin{table}[!htb]
\centering
\caption{Actual endpoint quality coverage of the unchanged all-state front,
averaged over refitted bootstrap draws. A tolerance comparable to sampling
uncertainty closes most of the exact-membership gap.}
\label{tab:l8_epsilon}
\small
\begin{tabular}{@{}lrrrr@{}}
\toprule
Predictor & $0$ SE & $0.5$ SE & $1$ SE & $2$ SE \\
\midrule
Exact-singleton additive & 81.93\% & 94.35\% & 99.44\% & 100.00\% \\
37-anchor quadratic & 98.03\% & 99.99\% & 100.00\% & 100.00\% \\
\bottomrule
\end{tabular}
\end{table}

We also fix $u(x)=\sum_j f_j(x)/\tau_j$, where $\tau_j$ is the sample standard
deviation over the 37 original anchors. Both coefficients are one, and the
scales remain fixed during resampling. Selection minimizes predicted $u$
over all 256 states without a feasibility constraint. All four predictors
select mask 5 on the original sample. Holding that state fixed gives the same
mean regret, $0.1048$, for all models. Recomputing prices and selecting anew
gives mean exact regrets $0.05843$ (singleton additive), $0.0002953$ (anchored
quadratic), $0.02250$ (support-$\le3$ additive) and $0.0001588$
(support-$\le3$ quadratic). Thus the improvement concerns sensitivity of the
selection to resampled measurements, not a difference between the original
fixed outputs.

\paragraph{Decisions among unqueried combinations.}
The all-state utility optimum on the original sample, mask 5, is a measured
pair anchor. To isolate decisions beyond measured states, we additionally
restrict \emph{both} candidate selection and the oracle comparator to
$\mathcal U=\{x:\lVert x\rVert_0\ge3\}$, the 219 unqueried combinations.
We recompute the true and predicted fronts within $\mathcal U$ rather than
intersecting the all-state fronts with it. The predictors, 2,000 resamples,
normalization and tolerance scales remain unchanged. Regret is
$u(\widehat x)-\min_{x\in\mathcal U}u(x)$, where $\widehat x$ minimizes
predicted utility within $\mathcal U$; no feasibility constraint is imposed.
This domain-restricted analysis is post hoc on the same historical FIT units.

\begin{table}[!htb]
\centering
\caption{Refitted decisions restricted to 219 unqueried combinations.
Coverage evaluates the actual endpoints returned by the predicted front.
Regret uses the fixed normalized utility and an oracle restricted to the
same candidate set. All entries are means over paired resamples.}
\label{tab:l8_unqueried_decisions}
\footnotesize
\setlength{\tabcolsep}{4pt}
\begin{tabular}{@{}lrrrrr@{}}
\toprule
Predictor & Recall (\%) & Precision (\%) & $0.5$-SE coverage & $1$-SE coverage & Regret \\
\midrule
Exact-singleton additive & 79.33 & 79.24 & 94.46\% & 99.51\% & 0.064898 \\
37-anchor quadratic & 97.42 & 98.20 & 99.99\% & 100.00\% & 0.000365 \\
\bottomrule
\end{tabular}
\end{table}

Both predictors choose the three-move mask 133 on the original sample.
The regret difference therefore concerns re-selection after resampling, as
in the full-domain comparison. The $95\%$ variability intervals for
restricted-front recall are $[50,100]\%$ for addition and $[80,100]\%$ for
the quadratic. In the full domain, $53.95\%$ and $66.50\%$ of their refitted
choices, respectively, are among the 37 measured anchors. Restricting the
domain removes these choices while retaining the decision-stability gap.

The fixed scales, in nats, are
\[
 \begin{aligned}
 (\sigma_{\mathrm{NLL}},\sigma_{\mathrm{ARC}})
  &=(0.000905588,\;0.003973544),\\
 (\tau_{\mathrm{NLL}},\tau_{\mathrm{ARC}})
  &=(0.000874855,\;0.012814780).
 \end{aligned}
\]
\section{The 4B Solver Lattices and the Grid-Snap Cells}
\label{app:bridge}
\label{sec:bridge}

The reconstruction objective improves along two operations on the grid-trained
master: refitting free magnitudes with codes fixed, then re-selecting codes
by exhaustive magnitude-ranked reconstruction minimization. A separate path
snaps the classic arm's recovered free levels to the NVFP4 grid while keeping
its code planes. Every combination of eight metadata-selected dose batches
is evaluated for each operation.

\begin{figure}[htb]
\centering
\includegraphics[width=\textwidth]{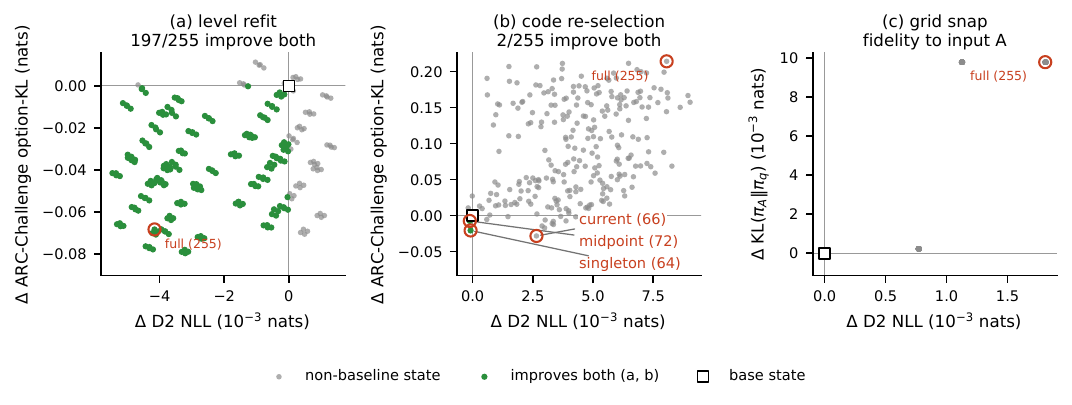}
\caption{Exact effects of three Qwen3-4B operations.
Each point is one of 255 non-baseline subsets. In (a,b), KL uses the
grid-trained master as reference; green points improve both functionals
(197 for level refit, two for code re-selection).
Panel (c) measures fidelity to the input A itself: its KL is minimized at A,
so this panel shows preservation costs, on an ARC axis $10^3$ finer.
The re-selection panel marks the full state, two additive selections and the
improving singleton.}
\label{fig:bridge}
\end{figure}

The full eight-batch level refit lowers D2 by $0.0042$ and ARC-KL by $0.068$;
the full re-selection raises them by $0.0081$ and $0.214$. The latter
lattice contains one improving singleton and one improving pair.
Current-state additive pricing gets $2.75\%$ of its D2 signs correct with 33
false-safe predictions; own-midpoint singleton pricing gets $99.6\%$ with
one, and selects the improving pair. These operations expose a failure of
isolated pricing as well as the interactions measured in the main text.

\begin{table}[htb]
\centering
\caption{Three exhaustive $256$-state lattices of eight manifest-selected dose
batches on Qwen3-4B under D2 NLL and ARC-Challenge option-KL.
Common descent counts jointly negative endpoint
differences; false-safe means predicted improvement on both but exact worsening
on one. These columns compare the two ladder operations against their fixed
master. Grid-snap KL instead measures fidelity to its own input A, so the
joint-improvement columns do not apply. D2 sign accuracy is reported for all
255 states; resolved-state counts for grid snapping are given below.}
\label{tab:bridge}
\vspace{0.3em}
\footnotesize
\setlength{\tabcolsep}{4pt}
\begin{tabular}{@{}lccc@{}}
\toprule
& common & D2 sign, all 255 & false-safe \\
Move & descent & current / midpoint & current / midpoint \\
\midrule
level refit & 197/255 & 85.1 / 97.3\% & 40 / 1 \\
code re-selection & \textbf{2/255} & \textbf{2.75} / \textbf{99.6\%} & 33 / 1 \\
grid snap (self-ref.\ KL) & --- & 48.2 / 98.4\% & --- \\
\bottomrule
\end{tabular}
\end{table}

The 4B lattices have two starting states of Qwen3-4B. The level-refit and
re-selection lattices start from the grid-trained master's constrained hard
state: an E2M1 magnitude pair $\times$ UE4M3 scale per group of 128, with the
current assignment. The grid constrains decoded values under this metadata;
it does not use the native NVFP4 block-16 storage format or kernels. Refitting
re-solves each group's two magnitudes with codes fixed. Code re-selection
then exhaustively searches magnitude-ranked assignments with free levels.
Reconstruction decreases along both steps.

The grid-snap lattice starts from the classic free-codebook arm's 1600-step
endpoint (cell A below) and preserves its code planes when snapping onto the
grid (B$_{\mathrm{plane}}$). The historical projection re-ran Lloyd's algorithm
before snapping (B$_{\mathrm{legacy}}$); the comparison below separates these
operations.

\begin{table}[htb]
\centering
\caption{Endpoint cells on Qwen3-4B (held-out WikiText perplexity), reproduced
in the Exp.\ C$'$ recovery.  A: free two-bit codebook recovered for 1600 steps.
B$_{\mathrm{plane}}$: A's code planes kept, codebook snapped to the NVFP4 grid
(code-plane reproduction error $0$).  Lloyd$_{\mathrm{free}}$: Lloyd reallocation and free-level refit
($9{,}736{,}593$ strong-plane assignments changed, maximal requantization error
$1.0744$). B$_{\mathrm{legacy}}$ then snaps that reassigned codebook.}
\label{tab:cells}
\vspace{0.4em}
\small
\begin{tabular}{@{}llrr@{}}
\toprule
Cell & Construction & PPL & vs.\ A \\
\midrule
A & free endpoint & 20.9293 & --- \\
B$_{\mathrm{plane}}$ & grid snap of A, codes fixed & 21.1049 & $+0.84\%$ \\
Lloyd$_{\mathrm{free}}$ & reallocation before grid snap & 36.3365 & $+73.62\%$ \\
B$_{\mathrm{legacy}}$ & Lloyd re-assignment of A, then snap & 36.1654 & $+72.80\%$ \\
\bottomrule
\end{tabular}
\end{table}

On the log-perplexity scale B$_{\mathrm{plane}}$ accounts for $1.53\%$ of the
A$\to$B$_{\mathrm{legacy}}$ damage along the path that preserves assignments.
The reassigned free-level state before the snap has PPL $36.336523$:
$+73.62\%$ relative to A. Snapping that state changes PPL by $-0.171123$
($-0.47\%$), so the large loss is already present after Lloyd reassignment.
All four cells use the same 65,536-token held-out stream.

\paragraph{Grid-snap fidelity lattice.}
The bank covers the $28{,}385{,}280$ groups whose codebook changes under the
snap, split into 26 dose batches by positive proxy cost; the L8 selected batches
$\{3,6,9,13,15,17,22,25\}$ before any functional was read.
Its reference is the classic deployment state A itself, reproduced exactly by
the left compact artifact. Thus $F_A(q)=\mathrm{KL}(\pi_A\Vert\pi_q)$ has
$F_A(A)=0$, $F_A(q)\ge0$ and $\nabla F_A(A)=0$.
This coordinate measures preservation of A's option distribution, with a
different reference from the ladder experiments.
Full snap: D2
$+0.00180859$ (CI $[-0.00083762,+0.00419934]$), ARC option-KL $+0.00978934$
(CI $[+0.00827162,+0.01157384]$, all $586$ item deltas positive).
Degree-two energy is $0.1165\%$ (D2) and $0.0134\%$ (ARC);
degree-three-and-higher is below $10^{-6}$.
Over all 255 states, D2 sign agreement is $48.2\%$ for current-state additive
prices, $98.4\%$ for singleton-midpoint prices and $98.8\%$ for exact-singleton
prices. Current and midpoint both agree on all 103 D2 endpoints resolved by
paired $95\%$ intervals ($10{,}000$ bootstrap replicates).
Sums of the 28 pair coefficients:
D2 $-9.1\times10^{-5}$ $[-1.46\times10^{-4},-3.6\times10^{-5}]$, ARC
$-2.2\times10^{-4}$ $[-4.8\times10^{-4},+2.3\times10^{-5}]$.

\paragraph{Numerical zero in self-referenced KL.}
The saved FP32 baseline KL is $-1.2794\times10^{-9}$; the largest absolute
per-item baseline residual is $1.2244\times10^{-7}$.
Mask 6 (batches $\{6,9\}$) has raw differences $-3.6982\times10^{-8}$ D2 and
$-1.3933\times10^{-9}$ KL; mask 13 has KL difference $-5.7720\times10^{-11}$.
These negative KL differences are numerical residuals, not improvements over
the self-reference minimum. They remain unchanged in the measurement files
and are not marked as jointly improving states in Figure~\ref{fig:bridge}.
The eight stored current-state directional responses have maximum absolute
mean $1.1453\times10^{-9}$, although each is analytically zero.
The generic sign counter's $10^{-12}$ threshold yields $192/255$ ($75.3\%$)
agreement from these residuals. This is a numerical-zero diagnostic, not an
ARC predictor result. Since nonnegative ARC changes are known from the
functional definition, ARC sign-agreement and joint false-safe counts are
excluded from the predictor comparison for this lattice; the D2 comparison
and measured KL magnitudes remain informative.

\paragraph{Solver-step lattices.}
Both lattices use the same manifest-only selection rule on the 26 dose batches
of each step (level-refit batches $\{1,4,8,12,16,19,22,24\}$, re-selection
batches $\{3,7,10,12,14,18,22,23\}$); the 26 known batch outcomes were not
consulted, and the selected batches are representative of them (refit: 3/8
selected and 6/18 omitted batches harm D2; re-selection: 7/8 and 15/18).
Functionals: D2 odd NLL on 28 held-out blocks and ARC-Challenge option-KL on 586
label-free development items; 18 centers (zero, eight unit, full, four
antithetic pairs).  Exact fronts: refit 10 states, 197 common-descent (3
singletons, 194 compositions); re-selection 4 states $\{8,64,66,72\}$,
common-descent $\{64,72\}$ only.  Full quadratic $R^2$: $0.999999/0.999978$
(refit D2/ARC), $0.999342/0.997375$ (re-selection); degree-two energy
$0.014\%/0.235\%$ and $2.413\%/4.560\%$.  Sums of the 28 pair coefficients:
refit ARC $+0.0104$ $[+0.0083,+0.0125]$ (harmful), the other three unresolved.
Rank-4 truncation of the 18-center operator on the re-selection D2 falls to
$53.9\%$ sign agreement, consistent with the center dependence of
Appendix~\ref{app:atlas}.  Peak memory $75$ GiB (refit) and $60$ GiB
(re-selection) on one RTX PRO 6000.

\section{Selection Runs: Protocols and Results}
\label{app:protocols}

\begin{table}[htb]
\centering
\caption{Qwen3-4B code re-selection repair. Changes from input B, with paired
$95\%$ intervals. Continuous TEST: seven NLL blocks and 154 ARC-Challenge
items, with ARC intervals resampling ten stored microbatch means.
Accuracy: 586 unused ARC-C items, 10,042 HellaSwag items and 14,042 MMLU items.
HellaSwag completion uses normalized accuracy.}
\label{tab:constructors}
\footnotesize
\setlength{\tabcolsep}{2.5pt}
\begin{tabular}{@{}lccc@{}}
\toprule
& additive & screened beam & direct beam \\
\midrule
cost (fwd.-eq.) & \shortstack{94 midpoint\\26 exact-singleton} & 199 (69 decisive) & 199 \\
batches selected & 9 & 6 & 8 \\
$\Delta$ NLL & \shortstack{$-0.007$\\$[-0.022,+0.007]$} & \shortstack{$-0.0005$\\$[-0.0073,+0.0068]$} & \shortstack{$-0.0004$\\$[-0.0070,+0.0078]$} \\
$\Delta$ ARC-C KL & \shortstack{$-0.211$\\$[-0.280,-0.139]$} & \shortstack{$-0.186$\\$[-0.251,-0.121]$} & \shortstack{$-0.226$\\$[-0.301,-0.155]$} \\
\midrule
ARC-C acc.\ (pp) & $+2.7$ $[+0.3,+5.1]$ & $+2.1$ $[-0.2,+4.3]$ & $+4.1$ $[+1.7,+6.7]$ \\
MMLU acc.\ (pp) & $+2.8$ $[+2.2,+3.4]$ & $+1.6$ $[+1.1,+2.2]$ & $+3.1$ $[+2.6,+3.8]$ \\
Hella label (pp) & $+2.9$ $[+2.1,+3.7]$ & $+0.8$ $[+0.1,+1.4]$ & $+2.4$ $[+1.7,+3.2]$ \\
Hella completion (pp) & $-0.18$ $[-0.57,+0.22]$ & $0.00$ $[-0.38,+0.38]$ & $-0.15$ $[-0.52,+0.23]$ \\
\bottomrule
\end{tabular}
\end{table}

The early composition studies price moves at a common center
and re-score selected candidates at their own midpoints. The later additive
and endpoint-beam comparisons directly contrast isolated prices with
state-dependent evaluation. Each protocol specifies its statistical units
and data-use history. Intervals use paired Student-$t$ or item/block bootstrap
over those units.

\paragraph{Evaluation units and their history.}
The 4B repair's FIT, VALIDATION and TEST units ($14+7+7$ D2 blocks,
$288+144+154$ ARC-Challenge items) partition the 28 blocks and 586 items on
which the lattices of Appendix~\ref{app:bridge} had already been evaluated.
After that evaluation, we fixed the constructor rules, read the fit partition,
selected one state on validation, and read the test partition once. TEST is
unused in the constructors' fit and selection but was seen earlier in the study.
The accuracy tasks of Appendix~\ref{app:accuracy}, the
586 ARC-Challenge items outside the repair, and the original Llama-3.2 outer
splits were read for the first time at their original readouts. Later matched
controls explicitly reuse those units. The recentered selection uses newly
opened C4 blocks and ARC-Easy items, as specified below.

\subsection{Common-center composition studies}

\paragraph{Composition does two different jobs.}
The two selected ARC 4B components have opposite functional effects.
Their pair improves both point estimates at exact endpoints
(Table~\ref{tab:arcpair}), with ARC-Easy accuracy $491\to505/594$ against $501$
and $488$ for its components. The pair's D2 interval spans zero.
On the 0.6B HellaSwag bank, one primitive already improves all three declared
functionals and beats the selected pair on the target
($-0.078$ against $-0.065$ KL), while the pair is better on both generic domains
(Table~\ref{tab:components}). Composition can therefore add feasible
trade-offs or redistribute improvement among functionals.
The common-center composition studies also measure the center shift directly:
on the 1.7B CommonsenseQA bank, own-midpoint rescoring reversed all three
candidates that passed common-center confirmation. The corresponding target
endpoints remained unread; WinoGrande also returned no certified state.

\paragraph{ARC 4B (development).}
\label{app:arc_history}
$\funcset=\{$D2 NLL, ARC-Easy option-KL$\}$, 26 primitives.  Of 16 pairs passing
a two-means-and-two-halves screen, $\{5,17\}$ had the largest weakest
standardized margin ($3.23$ against $1.36$).  Materialized state: $65{,}552$
groups (16 from primitive 5, $65{,}536$ from 17), all 252 projections decoding
exactly.  Own-midpoint rescoring: D2 $-0.00457/-0.00458/-0.00458$
(coarse/fine/Richardson, 22/28 units negative), target
$-0.04429/-0.04441/-0.04445$ (53/74).  Held-out: D2 odd NLL
$2.99799\to2.99672$; ARC-Easy $491\to505/594$ (17 rescues, 3 regressions,
McNemar $p=0.00258$); gold margin $+0.282$ [$+0.234,+0.328$]; top-1 margin
$+0.192$; option CE $-0.0553$; teacher-KL $-0.0426$.  Undeclared D5: predicted
$+0.0029088$ from the response matrix, observed $+0.0029598$.  The D2 outer
upper bound was $+0.00027$; the exact component controls of
Table~\ref{tab:arcpair} were run after the pair was frozen.

\begin{table}[htb]
\centering
\caption{ARC 4B pair and its components at exact endpoints, paired $95\%$
intervals.  Negative is improvement.}
\label{tab:arcpair}
\vspace{0.4em}
\footnotesize
\setlength{\tabcolsep}{3.5pt}
\begin{tabular}{@{}lrrr@{}}
\toprule
State & D2 NLL & ARC option-KL & acc.\ /594 \\
\midrule
primitive 5  & $+0.00167$ $[+0.00027,+0.00313]$ & $-0.04951$ $[-0.06136,-0.03742]$ & 501 \\
primitive 17 & $-0.00259$ $[-0.00533,-0.00019]$ & $+0.01672$ $[+0.00290,+0.03068]$ & 488 \\
pair $\{5,17\}$ & $-0.00127$ $[-0.00421,+0.00129]$ & $-0.04264$ $[-0.05880,-0.02653]$ & \textbf{505} \\
\bottomrule
\end{tabular}
\end{table}

\paragraph{TruthfulQA 4B (pre-registered).}
With three functionals and 26 primitives, the minimum-cardinality selector
finds one fit-feasible singleton (batch 8, 128 groups). Validation responses
are D2 $-0.00284$, D3 $-0.00321$, and target $-0.00729$, with target halves
$-0.00888$ and $-0.00539$. The target's
standardized margin $1.66$ fell short of the Bonferroni critical value $2.47$
(UCB $+0.00353$), so no state was materialized and the 171 evaluation items were
not scored.

\paragraph{HellaSwag 0.6B (consensus protocol).}
This run uses 23 primitives, three functionals and five-fold consensus,
requiring the same candidate in four folds. Support 1 reached 5/5 on $\{14\}$ (pooled out-of-fold weakest
margin $6.76$); support 2 reached 4/5 on $\{4,14\}$ ($7.78$); no support from
3 to 8 reached four folds.  Outer confirmation: $-0.0292$, $-0.0404$, $-0.0395$ with
upper bounds $-0.0230$, $-0.0339$, $-0.0236$; own-midpoint rescoring negative on
27/28, 37/38 and 48/50 units.

\begin{table}[htb]
\centering
\caption{HellaSwag components at exact endpoints (0.6B, 1{,}820 FINAL items for
the target columns).  All quantized states answer the same 424 items correctly;
the teacher answers 631.}
\label{tab:components}
\vspace{0.4em}
\small
\begin{tabular}{@{}lrrrr@{}}
\toprule
State & D2 NLL & D3 NLL & target KL & gold margin \\
\midrule
primitive 4 (8 groups) & $-0.00817$ & $-0.00632$ & $+0.01669$ & $-0.02916$ \\
primitive 14 (8{,}192) & $-0.01714$ & $-0.03347$ & $\mathbf{-0.07831}$ & $\mathbf{+0.10801}$ \\
pair (8{,}200) & $\mathbf{-0.02489}$ & $\mathbf{-0.03911}$ & $-0.06458$ & $+0.07488$ \\
\bottomrule
\end{tabular}
\end{table}

Pair intervals on FINAL: target KL $[-0.0705,-0.0587]$, gold margin
$[+0.0610,+0.0887]$, option CE $-0.0713$ $[-0.0825,-0.0599]$, top1--top2 gap
$-0.1730$ $[-0.1837,-0.1624]$; endpoint interactions relative to component sums
$+0.00043$, $+0.00068$, $-0.00296$, $-0.00398$.

\subsection{Qwen3-4B code re-selection repair}

\paragraph{4B re-selection bank: one-shot midpoint selection.}
This control uses the same bank, splits and functionals as the beam below.
Each of the 26 singleton batches receives one hooked backward at its own
midpoint on FIT, reusing rows already measured by the beam's anchor jobs.
The selector uses only the $i$-th coordinate of the $i$-th row and sums
these prices. It minimizes predicted ARC-Challenge option-KL subject to
predicted $\Delta$D2 $\le0$. Before reading any endpoint, it selects
$\{0,1,5,6,15,17,21,22,24\}$
(additive FIT prediction $-0.00041$, $-0.134$).  Exact changes from B: FIT
$-0.00139$ / $-0.1098$, VALIDATION $-0.00445$ / $-0.1175$, TEST $-0.00728$
$[-0.02175,+0.00663]$ / $-0.2112$ $[-0.2798,-0.1389]$.  TEST paired against the
solver's output C: $-0.0571$ $[-0.0837,-0.0327]$ / $-0.1659$
$[-0.2239,-0.1057]$; against the beam's primary: $-0.0068$ $[-0.0188,+0.0025]$
/ $-0.0256$ $[-0.0533,+0.0003]$ ($10{,}000$ paired draws; D2 over seven blocks,
ARC over ten stored microbatch means of 154 items).  Cost: 26 own-midpoint jobs,
each $3.63$ endpoint forwards at the median ratio measured on the beam's 36
anchor jobs, against the beam's 199 forward-equivalents.

\paragraph{Exact-singleton additive control on the 4B bank.}
The first 26 queries of the direct endpoint trace cover all singleton batches.
Applying the same additive selection rule to their exact FIT effects selects
$\{0,1,5,6,15,17,21,22,24\}$, exactly the midpoint-selected state.
The additive prediction is $-0.000682536$ D2 and $-0.123587976$ ARC option-KL.
Both price sources therefore share the same materialized state and all its
FIT, VALIDATION and TEST results. The exact-singleton price budget is 26
endpoint queries; all reused queries are charged.
This comparison isolates price measurement from the additive selection rule.

\paragraph{4B re-selection bank: recentered exact-endpoint beam.}
The beam starts at B (free levels, codes fixed) and searches the 26 dose
batches of the re-selection step under D2 odd NLL and ARC-Challenge option-KL.
FIT uses 14 blocks and 18 microbatches of 16 items; VALIDATION uses 7 blocks
and 144 items; TEST uses 7 blocks and 154 items, all disjoint. The beam has
width 4, permits additions only and uses slack 0 relative to B. At each
state, a cheap screen (current row, proxy, structural, deterministic random)
selects four candidates. Each receives an own-midpoint row and an exact FIT
endpoint. The persistent Pareto beam re-anchors at exact states and stops
when the frontier is exhausted, after nine rounds.

Five exact Pareto states are found on FIT, from which a four-state slate is
frozen. The validation policy, fixed before reading, minimizes ARC option-KL
subject to $\Delta$D2$\le0$ and selects
$\{6,10,15,17,21,24\}$.  TEST: D2 $3.586104\to3.585646$ ($-0.000458$, paired
$95\%$ $[-0.007320,+0.006771]$, 4/7 blocks negative); ARC option-KL
$0.759318\to0.573742$ ($-0.185575$, $[-0.250825,-0.121108]$, 10/10 microbatches
negative).

In a replay audit, zeroing or negating all 36 midpoint rows changes none of
the nine round decisions or the final frontier. The 36 anchor jobs
($10.4$ h) produced the $72/72$ versus $32/72$ sign diagnostic but did not
affect selection. The median own-midpoint-to-endpoint job cost ratio is
$3.63$. Of the total $199$ forward-equivalents, the 36 exact FIT endpoints
and nine current-state rows (about 69) entered decisions; the 36 anchor jobs
(about 131) did not.

A separate local-pricing replay uses the nine archived base states and the
same four measured candidates, scales, feasibility and tie rules.
Own-midpoint pricing reproduces the exact
local beam head $9/9$ with $0/36$ false-feasible candidate occurrences, the
recentered current row $5/9$ with $16/36$; in the four disagreements (rounds 4,
6, 7, 9) exact and midpoint pricing keep the incumbent where the current row
expands.

\paragraph{4B re-selection bank: direct endpoint control.}
This control keeps the bank, splits, beam semantics, slate rule and validation
policy. It spends the frozen budget of $199$ FIT evaluations on exact
endpoints, without current-state or midpoint scoring, and selects
$\{0,1,5,6,15,17,21,24\}$.
Changes from B: FIT $-0.00260$ / $-0.1166$, VALIDATION $-0.00643$ / $-0.0861$,
TEST $-0.00037$ $[-0.00697,+0.00780]$ / $-0.2262$ $[-0.3013,-0.1554]$.  TEST
paired against C: $-0.0502$ $[-0.0792,-0.0278]$ / $-0.1810$ $[-0.2525,-0.1073]$;
against the current-screened beam's primary: $+0.0001$ $[-0.0040,+0.0039]$ /
$-0.0407$ $[-0.0585,-0.0217]$; against the one-shot state: $+0.0069$
$[-0.0004,+0.0166]$ / $-0.0150$ $[-0.0429,+0.0109]$.  On VALIDATION the one-shot
state had the lower ARC option-KL (direct minus one-shot $+0.0314$
$[+0.0042,+0.0597]$); the frozen selection was not revisited.  Search wall-clock
$64.7$ min on one GPU with a persistent evaluator and cached decoded bank, not
comparable to the historical run's timing.

\paragraph{Matched 94-query direct prefix.}
A prefix of the same direct trace stops at 94 charged endpoint queries,
approximately the midpoint one-shot pricing budget.
The three FIT-Pareto candidates have the following VALIDATION effects:

\begin{table}[htb]
\centering
\caption{The 94-query direct prefix on the 4B bank. Every candidate fails the
fixed point constraint $\Delta$D2$\le0$, so no TEST endpoint is read.}
\label{tab:direct94}
\small
\begin{tabular}{@{}lrr@{}}
\toprule
batches & $\Delta$ D2 NLL & $\Delta$ ARC option-KL \\
\midrule
$\{0,1,15,21\}$ & $+0.00296671$ & $-0.04428619$ \\
$\{0,1,21\}$ & $+0.00017549$ & $-0.04773489$ \\
$\{0,1,6,15\}$ & $+0.00004282$ & $-0.04999044$ \\
\bottomrule
\end{tabular}
\end{table}

The outcome is abstention under the original rule, without a new tolerance.
At this query budget the beam and additive selector reach different candidates;
the exact-singleton additive control above reaches the same output as midpoint
selection in 26 queries.

\subsection{Task accuracy and scoring protocols}

\paragraph{Zero-shot accuracy of the frozen states.}
\label{app:accuracy}
All states are frozen before accuracy is read on ARC-Challenge test
($1{,}172$ items), HellaSwag validation ($10{,}042$) and MMLU test ($14{,}042$).
We use zero-shot restricted candidate-label likelihood: the options are shown,
and the model's most likely label is its answer. For HellaSwag, this selects
a label over the four shown endings, rather than scoring their
length-normalized completion as in the evaluation harness. Qwen3-4B uses its
direct-answer chat template with thinking disabled; Llama uses a plain prompt.
We store each item's prediction, gold margin and restricted-option KL to the
FP master. Differences use $10{,}000$ paired item bootstrap draws and exact
McNemar tests. Of the
ARC-Challenge items, $586$ had served as the 4B repair's FIT, VALIDATION and TEST
option-KL units; the main text reports the other $586$.

In ARC-C, HellaSwag and MMLU order, Qwen3-4B accuracies are
$85.32/68.74/63.37$ for FP, $66.13/48.73/44.93$ for B,
$66.55/41.83/44.17$ for C, $68.77/51.64/47.73$ for one-shot,
$69.80/51.17/48.07$ for direct, and $68.43/49.50/46.52$ for the beam primary.
Changes from B on all $1{,}172$ ARC-C items:
C $+0.43$ $[-1.54,+2.39]$, one-shot $+2.65$ $[+0.94,+4.35]$, direct $+3.67$
$[+1.88,+5.46]$, beam $+2.30$ $[+0.68,+3.92]$; on the 586 unused items, C
$+0.68$ $[-1.88,+3.41]$ (McNemar $p=0.70$), one-shot $+2.73$ $[+0.34,+5.12]$
($p=0.040$), direct $+4.10$ $[+1.71,+6.66]$ ($p=0.0015$), beam $+2.05$
$[-0.17,+4.27]$ ($p=0.088$).  Option-KL differences from B are negative with
intervals excluding zero for all three repaired states on all three tasks.

Llama-3.2-1B: FP $36.18/27.30/36.24$; standard GPTQ($1\%$) $23.98/25.80/26.27$;
midpoint one-shot $25.60/25.52/26.53$, changes $+1.62$ $[-1.79,+4.95]$, $-0.28$
$[-1.49,+0.93]$, $+0.26$ $[-0.77,+1.26]$; option-KL changes $-0.026$, $-0.073$,
$-0.027$, all resolved.  Queue $4$ h $56$ min on one GPU.

\paragraph{Standard HellaSwag completion.}
The same five Qwen3-4B states are evaluated on the 10,042-item validation set
with the \texttt{lm-evaluation-harness} 0.4.13 task semantics, zero-shot and no
chat template. Raw \texttt{acc} uses summed continuation log likelihood.
The primary \texttt{acc\_norm} divides by processed-ending character length.
States are fixed before this evaluation.

\begin{table}[htb]
\centering
\caption{Standard HellaSwag completion on the frozen 4B states.
Absolute accuracies are percentages; changes from B are percentage points
with paired $95\%$ intervals.}
\label{tab:hella_standard}
\footnotesize
\setlength{\tabcolsep}{4pt}
\begin{tabular}{@{}lrrll@{}}
\toprule
state & acc & acc$_{\mathrm{norm}}$ & $\Delta$acc & $\Delta$acc$_{\mathrm{norm}}$ \\
\midrule
B & 45.210 & 59.649 & --- & --- \\
C & 44.732 & 58.803 & $-0.478$ $[-0.876,-0.089]$ & $-0.846$ $[-1.354,-0.339]$ \\
one-shot & 45.091 & 59.470 & $-0.119$ $[-0.458,+0.219]$ & $-0.179$ $[-0.568,+0.219]$ \\
direct primary & 45.001 & 59.500 & $-0.209$ $[-0.518,+0.100]$ & $-0.149$ $[-0.518,+0.229]$ \\
beam primary & 45.300 & 59.649 & $+0.090$ $[-0.209,+0.398]$ & $0.000$ $[-0.378,+0.378]$ \\
\bottomrule
\end{tabular}
\end{table}

Exact McNemar $p$-values for raw / normalized accuracy are
$0.0223/0.00109$ (C), $0.525/0.399$ (one-shot),
$0.205/0.464$ (direct) and $0.615/1.000$ (beam).
C is harmful under both scoring rules. The repaired states have no resolved
improvement over B on standard completion, whereas their label-choice gains
are reported above.
B and the beam both answer 5990 items correctly under normalized completion,
with 182 paired rescues and 182 regressions.

\subsection{Llama-3.2-1B composition and support-matched selection}

\paragraph{FIT decomposition of the Llama one-shot states.}
\label{app:pricing_l1}
With $p(S)$ the selector's price sum, $e(S)$ the sum of exact single-replacement
endpoints from the common start and $\Delta(S)$ the exact endpoint of the
composed state on the same FIT units, $e-p$ is the pricing error and
$\Delta-e$ the composition residual, matching Equation~\eqref{eq:decomp}.
Thus $\Delta-p=(e-p)+(\Delta-e)$ on FIT. (NLL, ARC option-KL): midpoint
state $\Delta=(-0.106,+0.012)$, $p=(-0.0003,-0.452)$, $e=(-0.023,-0.447)$,
composition $(-0.083,+0.459)$, pricing $(-0.023,+0.005)$; current state
$\Delta=(+0.760,+0.002)$ (NLL paired CI $[+0.408,+1.165]$), $p=(-0.802,-0.782)$,
$e=(+1.032,-0.370)$, composition $(-0.272,+0.373)$, pricing
$(+1.835,+0.412)$; exact-singleton state $\Delta=(-0.066,+0.051)$,
$p=e=(-0.0005,-0.473)$, composition $(-0.066,+0.524)$.

The $0.00544884$ signed ARC pricing error of the 61-change midpoint state
can contain cancellation. Its absolute-error sum is $0.01772980$, mean
absolute error $0.00029065$, and largest error $0.00374947$.
The measured composition residual, $0.45883472$, is $25.88$ times that
absolute-error sum. Table~\ref{tab:llama_oneshot} reports the original
held-out constructions; the new support controls below isolate the effect
of changing fewer projections.

\begin{table}[htb]
\centering
\caption{Original Llama-3.2-1B constructions, INT3 g128. Each projection keeps
GPTQ($1\%$) or takes RTN or GPTQ($0.1\%$), including its scale and zero.
TEST: 114 blocks and 297 ARC-Easy items; paired $95\%$ intervals.
The exact-endpoint beam uses the midpoint-pricing budget of 352 forwards.}
\label{tab:llama_oneshot}
\footnotesize
\setlength{\tabcolsep}{3.5pt}
\begin{tabular}{@{}lccc@{}}
\toprule
Constructor & Changes & $\Delta$ NLL & $\Delta$ ARC option-KL \\
\midrule
\multicolumn{4}{@{}l}{\textit{GPTQ without activation ordering}} \\
Additive, midpoint & 61 & $-0.117$ $[-0.223,-0.020]$ & $+0.005$ $[-0.019,+0.030]$ \\
Additive, current & 86 & $+0.451$ $[+0.262,+0.659]$ & $+0.011$ $[-0.014,+0.035]$ \\
Additive, exact singleton & 73 & $-0.109$ $[-0.222,-0.003]$ & $+0.044$ $[+0.017,+0.071]$ \\
Endpoint beam & 2 & $-0.206$ $[-0.260,-0.156]$ & $-0.029$ $[-0.041,-0.017]$ \\
\midrule
\multicolumn{4}{@{}l}{\textit{GPTQ with activation ordering}} \\
Additive, midpoint & 37 & $+0.033$ $[+0.006,+0.062]$ & $-0.000$ $[-0.016,+0.014]$ \\
Endpoint beam & 2 & $-0.001$ $[-0.004,+0.001]$ & $-0.003$ $[-0.008,+0.001]$ \\
\bottomrule
\end{tabular}
\end{table}

\paragraph{Composition along a fixed support path.}
\label{app:support_path}
The 61 selected replacements are ordered by a deterministic hash permutation
before reading new composed endpoints. Prefixes of lengths 8, 16, 32 and 48
form one nested path. With $\Delta(S)$ the exact FIT effect, the midpoint
residual is $\Delta(S)-\sum_{i\in S}\hat s_i$ and the exact-singleton residual
is $\Delta(S)-\sum_{i\in S}s_i$.

\begin{table}[htb]
\centering
\caption{FIT effects along the fixed Llama support path.
Residuals are for ARC option-KL. Exact-singleton residuals retain the
composition error after isolated pricing error is removed.}
\label{tab:support_path}
\small
\begin{tabular}{@{}rrrrr@{}}
\toprule
support & exact $\Delta$NLL & exact $\Delta$ARC & midpoint residual & exact-singleton residual \\
\midrule
8 & $-0.050717$ & $-0.042166$ & $+0.042399$ & $+0.038998$ \\
16 & $-0.042965$ & $-0.038350$ & $+0.084378$ & $+0.080173$ \\
32 & $-0.044078$ & $+0.052146$ & $+0.268734$ & $+0.260388$ \\
48 & $-0.044561$ & $+0.049290$ & $+0.399687$ & $+0.391219$ \\
\bottomrule
\end{tabular}
\end{table}

The two smaller prefixes improve both FIT functionals; by support 32 the ARC
effect is harmful. Replacing midpoint prices with exact-singleton prices leaves
almost the same error. These residuals include all interactions of order two
and higher. The observed transition belongs to this fixed path, rather than
to a distribution of randomly chosen subsets.

\paragraph{Llama-3.2-1B matched exact-endpoint beam.}
Budget frozen before search as the one-shot midpoint pricing time divided by
the archived mean exact-singleton time, $\lfloor5{,}274.0/14.96\rfloor=352$
endpoint forwards, all charged (the 224 singleton answers were reused where the
search requested them; 128 composed states were materialized and scored).  A
width-4 exact Pareto beam from the GPTQ($1\%$) model exhausted the budget after
one expansion with a three-state slate of support two; the frozen validation
rule selected $\{9\to\mathrm{RTN},\,12\to\mathrm{GPTQ}(0.1\%)\}$.  VALIDATION:
NLL $3.0975\to2.9314$ ($-0.1661$ $[-0.2118,-0.1240]$), ARC $0.2347\to0.1995$
($-0.0352$ $[-0.0492,-0.0219]$); TEST: $3.2056\to2.9993$ ($-0.2063$
$[-0.2603,-0.1555]$), $0.2110\to0.1821$ ($-0.0289$ $[-0.0407,-0.0173]$).  Against
whole RTN on TEST: $-0.4929$ / $-0.0753$; against whole GPTQ($0.1\%$):
$-0.5862$ / $-0.0130$ $[-0.0331,+0.0075]$.  Against the 61-change midpoint
state, same units: VALIDATION $-0.0789$ $[-0.1047,-0.0518]$ / $-0.0355$
$[-0.0602,-0.0109]$, TEST $-0.0891$ $[-0.1645,-0.0085]$ / $-0.0342$
$[-0.0574,-0.0112]$.  This control was run after the shared VALIDATION/TEST had
been read for the one-shot states; it was pre-registered and its rule did not
depend on them.

\paragraph{Support-matched exact-singleton selection.}
\label{app:support_matched}
The follow-up uses the original no-act-order bank and exact FIT singleton
prices, with caps $\{2,4,8\}$ fixed before any new combination endpoint is read.
Cap two is primary; four and eight are sensitivity checks. The original
objective minimizes additive ARC-KL subject to additive NLL non-increase.
Candidate identities, encodings and prices are unchanged. With zero-based
projection indices and alternative labels 1=RTN, 2=GPTQ($0.1\%$), the states are
\begin{align*}
 S_2 &= \{(2,2),(9,1)\},\\
 S_4 &= S_2\cup\{(44,2),(95,2)\},\\
 S_8 &= S_4\cup\{(3,2),(12,2),(48,1),(100,1)\}.
\end{align*}
The fixed beam state is $\{(9,1),(12,2)\}$.
The five states, including the base, are evaluated on FIT (28 blocks,
592 items), VALIDATION (124 blocks, 297 items) and TEST (114 blocks,
297 items). All three additive states pass the original validation
constraint $\Delta$NLL$\le0$. No cap is changed after TEST.
These are retrospective comparisons on already exposed units, not a
second independent selection trial.

Intervals in this follow-up use $10{,}000$ paired bootstrap draws with seed
20260913: token blocks for NLL and individual items for ARC. Table~\ref{tab:support_matched}
reports differences from the base. Table~\ref{tab:support_paired} gives direct
paired differences, which cannot be inferred by subtracting the intervals.

\begin{table}[htb]
\centering
\caption{Support-limited additive states minus the fixed endpoint beam on
the same TEST units. Negative favors the additive state. Cap two is the
equal-support comparison; four and eight measure sensitivity to its cap.}
\label{tab:support_paired}
\small
\begin{tabular}{@{}lrr@{}}
\toprule
Comparison & $\Delta$ NLL [95\% interval] & $\Delta$ ARC-KL [95\% interval] \\
\midrule
Cap 2 $-$ beam & $+0.09483$ $[+0.05205,+0.13543]$ & $-0.00886$ $[-0.02078,+0.00289]$ \\
Cap 4 $-$ beam & $+0.09478$ $[+0.05254,+0.13454]$ & $-0.01475$ $[-0.02778,-0.00131]$ \\
Cap 8 $-$ beam & $+0.03436$ $[-0.05530,+0.11799]$ & $-0.01384$ $[-0.02955,+0.00168]$ \\
\bottomrule
\end{tabular}
\end{table}

On FIT, the ARC singleton sums for caps 2/4/8 are
$-0.038491/-0.072722/-0.134751$; their exact composed effects are
$-0.037058/-0.043067/-0.047570$, leaving composition residuals
$+0.001433/+0.029656/+0.087181$. The beam's sum is $-0.036648$, exact
effect $-0.031605$ and residual $+0.005043$. Small support controls much of
the error, while these states allocate improvement differently across NLL
and ARC. Intervals are per comparison, not simultaneous safety certificates.

\subsection{Recentered selection on fresh evaluation units}
Starting from the 61-change midpoint state, each round prices all 224
one-batch alternatives at their own midpoints on the development FIT units.
The selector makes at most 16 categorical changes, minimizing summed ARC
price subject to summed NLL price $\le0$, and advances to the predicted state
without veto. Exact FIT deltas per round (NLL, ARC) are $(-0.0057,-0.0498)$,
$(+0.0095,-0.0013)$, $(-0.0039,+0.0639)$ against predictions
$(-0.0002,-0.286)$, $(-0.0018,-0.182)$, $(-0.0015,-0.181)$.  Price drift over the
208 moves that remain non-current at consecutive centers: NLL Spearman
$0.937$ and $0.923$ with $25/208$ and $25/208$ sign flips; ARC Spearman $-0.247$
and $-0.091$ with $121/208$ and $105/208$.

Fresh evaluation uses 56 VALIDATION and 56 TEST blocks of 2048 tokens from
a newly downloaded C4 validation shard, with zero block-hash overlap with any
prior unit, and two unopened ARC-Easy sets of 594 items each. The frozen
validation rule minimizes ARC option-KL subject to point NLL no worse than
the start. The direct state and whole models serve only as controls. This
rule selects round 2. TEST losses are GPTQ($1\%$) $3.9677$ / $0.2353$; RTN $4.1357$ /
$0.2515$; GPTQ($0.1\%$) $4.4127$ / $0.2153$; direct primary $3.8064$ / $0.2063$;
one-shot start $3.8067$ / $0.2250$; rounds 1--3 $3.8371$ / $0.2013$, $3.8963$ /
$0.1791$, $3.9453$ / $0.2854$.  Round-2 primary against the start: $+0.0896$
$[+0.0180,+0.1837]$ / $-0.0460$ $[-0.0553,-0.0367]$; against the direct primary
$+0.0899$ $[+0.0134,+0.1765]$ / $-0.0272$ $[-0.0424,-0.0129]$; against
GPTQ($1\%$) $-0.0714$ $[-0.1685,+0.0314]$ / $-0.0562$ $[-0.0728,-0.0399]$.  The
direct primary against the start on these fresh units: $-0.0003$
$[-0.0812,+0.1073]$ / $-0.0187$ $[-0.0364,-0.0003]$.  Pricing cost $15{,}745$ s
for three rounds.

Figure~\ref{fig:reversal_all} plots both consecutive-center comparisons.
Sign-flip counts and rank correlations use prices measured at the two centers.
A move is comparable when it is a non-current alternative at both centers; the
16 batches changed between two centers take a different alternative and so a
different displacement.  Restricted to the 192 moves whose own projection did
not change, the ARC reversals are 117 and 99 (Spearman $-0.34$ and $-0.12$)
and the NLL reversals 20 and 20 ($0.97$ and $0.96$). These reversals persist
when the displacement is fixed and only its surrounding state changes.

\begin{figure}[htb]
\centering
\includegraphics[width=0.85\textwidth]{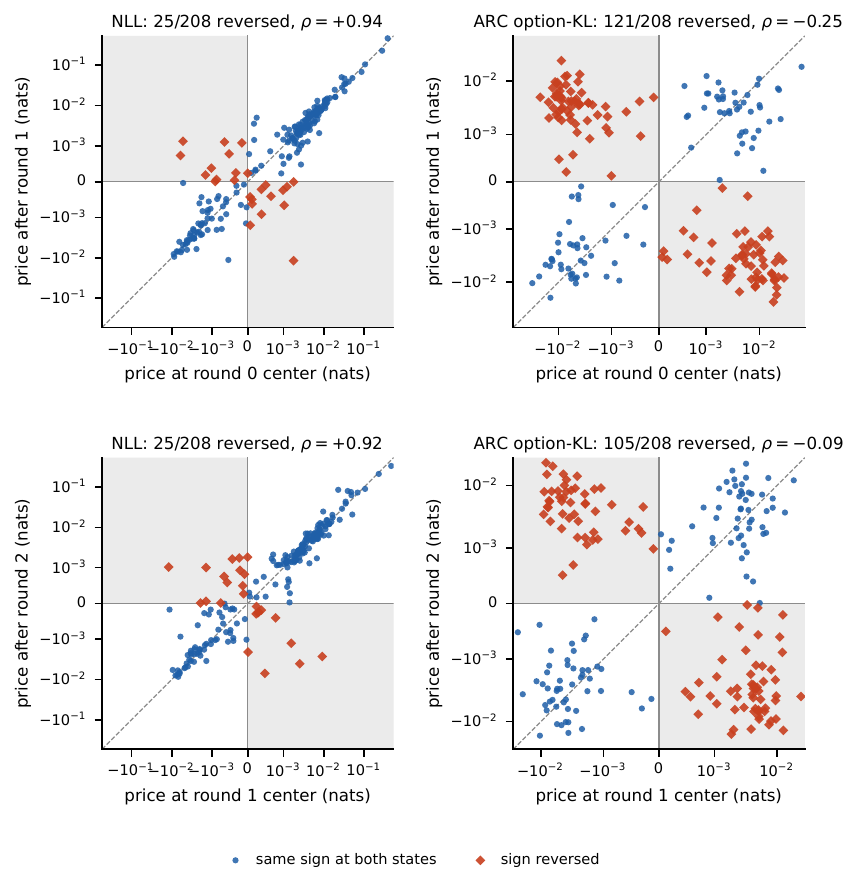}
\caption{Own-midpoint prices of the 208 comparable moves at consecutive centers
of the recentered selection (top: start to round 1; bottom: round 1 to round
2).  Symmetric-log axes; shaded quadrants are sign reversals.  NLL: 25 and 25
reversals, Spearman $0.94$ and $0.92$; ARC option-KL: 121 and 105 reversals,
Spearman $-0.25$ and $-0.09$.}
\label{fig:reversal_all}
\end{figure}

\subsection{Llama-3.2-3B configurations and composed endpoints}

\paragraph{Scale gate.}
A pre-registered rule capped FIT pricing at $24$ h.  Llama-3.1-8B fits the GPU
($76.7$ GB peak, producers $25$ min) but its 448 replacements project to $40.2$ h
($26.2$ midpoint, $14.0$ exact singleton), so the scale runs used Llama-3.2-3B
(28 layers, 196 projections, 392 replacements).

\paragraph{Llama-3.2-3B without activation ordering: a configuration failure.}
The first 3B run used the 1B configuration (no activation ordering).  Its
GPTQ($1\%$) model is far worse than RTN (TEST NLL: RTN $-3.586$
$[-4.593,-2.519]$ relative to it; ARC $-0.173$). The reference GPTQ
implementation reproduces this configuration exactly: zero decoded-weight
difference on all 196 projections and FIT NLL $7.053878$ in both. This
confirms a configuration failure on this model. We report the measurements
below separately from the scale comparison. FIT predictors against
exact single-replacement endpoints: midpoint $392/392$ NLL (Spearman $0.9996$)
and $390/392$ ARC ($0.9993$); current state $318/392$ ($0.835$) and $365/392$
($0.957$).  Selections: midpoint 130 batches, exact singleton 130 (4 differ),
current 137 (11 and 13 differ).  TEST against GPTQ($1\%$): midpoint $-3.712$
$[-4.719,-2.644]$ / $-0.041$ $[-0.127,+0.046]$, exact singleton $-3.712$ /
$-0.052$, current $-3.693$ / $-0.002$; midpoint against current $-0.0186$
$[-0.0219,-0.0151]$ / $-0.0391$ $[-0.0592,-0.0194]$, against RTN $-0.1255$
$[-0.1344,-0.1166]$ / $+0.1322$ $[+0.0910,+0.1738]$ (28 NLL blocks, 297 ARC
items per partition).

\paragraph{Llama-3.2-3B with activation ordering.}
GPTQ with \texttt{act\_order} uses static groups defined on the original
column order, keeping every state contiguous g128 INT3. With block size 128,
sequential calibration and $1\%$ or $0.1\%$ damping, it matches the reference
implementation to zero decoded-weight difference. Whole-model losses on the FIT units (NLL /
ARC option-KL): RTN $2.5130$ / $0.3948$; GPTQ($1\%$) without act-order $7.0539$
/ $0.5424$; act-order GPTQ($1\%$) $2.2808$ / $0.3695$; act-order GPTQ($0.1\%$)
$2.2607$ / $0.3915$.  Act-order GPTQ($1\%$) against RTN: NLL $-0.2322$
$[-0.2712,-0.1700]$, ARC $-0.0253$ $[-0.0542,+0.0034]$.

Each of the 196 projections keeps act-order GPTQ($1\%$) codes or takes RTN
or act-order $0.1\%$ codes. The 392 replacements are priced by the
current-state row, own midpoint and exact singleton endpoint on FIT.
At exact endpoints,
269/392 replacements raise NLL and 194/392 raise ARC option-KL (by source: RTN
152 and 100 of 196, $0.1\%$ 117 and 94); the two directions conflict on
$183/392$ (NLL improves and ARC worsens 54, NLL worsens and ARC improves 129,
both improve 69, both worsen 140), with Pearson $0.20$ and Spearman $0.03$
between the two coordinates.  Predictors: midpoint $390/392$ NLL (Spearman
$0.9996$, MAE $0.00037$, largest error $0.032$) and $391/392$ ARC ($0.9993$,
$0.00035$); current state $153/392$ ($0.292$, $0.0044$, largest error $0.59$)
and $357/392$ ($0.931$, $0.0040$); both directions jointly $389/392$ against
$141$, conflict classified on $389/392$ against $164$; by source the midpoint
is $195/196$ and $196/196$ (RTN) and $195/196$ on both ($0.1\%$).

The frozen
one-shot selectors (minimize predicted ARC subject to predicted NLL $\le0$)
change 149 (current), 132 (midpoint) and 134 (exact singleton) batches; the
midpoint and exact-singleton selections differ at 2 of 196 batches, the
current selection at 24 and 22, and summed over exact singleton endpoints the
current selection violates its NLL constraint ($+0.210$) before any composition
residual. The original pricing study stopped at selection; the endpoint
follow-up below executes those same states. Cost: current row $120$ s;
392 midpoints $23{,}279$ s ($59.4$ s each);
392 exact endpoints $15{,}497$ s ($39.5$ s each), ratio $1.50$.

\paragraph{Executing the original 3B act-order selections.}
\label{app:3b_composed}
The midpoint (132 changes), current (149) and exact-singleton (134) states
retain their original identities; no pricing or selection is repeated.
Each projection takes the complete codes, scale and zero tuple of its
selected alternative. Evaluation uses the original frozen manifests:
FIT has 28 blocks and 592 ARC-Easy items; VALIDATION and TEST each have
28 blocks and 297 items. The midpoint state remains the primary state.
Table~\ref{tab:3b_composed} gives TEST endpoint differences with the same
paired-bootstrap procedure as the support-matched follow-up.

\begin{table}[htb]
\centering
\caption{Actual high-support compositions on Llama-3.2-3B with
activation-ordered GPTQ. Differences from the GPTQ($1\%$) base on TEST.
All NLL increases are resolved; ARC directions are unresolved.}
\label{tab:3b_composed}
\small
\begin{tabular}{@{}lrrr@{}}
\toprule
Prices & Changes & $\Delta$ NLL [95\% interval] & $\Delta$ ARC-KL [95\% interval] \\
\midrule
Midpoint & 132 & $+0.06385$ $[+0.05732,+0.07102]$ & $-0.05165$ $[-0.10438,+0.00184]$ \\
Current & 149 & $+0.09419$ $[+0.07704,+0.12133]$ & $+0.01121$ $[-0.05079,+0.07573]$ \\
Exact singleton & 134 & $+0.06256$ $[+0.05567,+0.07001]$ & $-0.04079$ $[-0.09472,+0.01337]$ \\
\bottomrule
\end{tabular}
\end{table}

\begin{table}[htb]
\centering
\caption{Same-FIT ARC decomposition of the fixed 3B states.
The composition column is the exact combination effect minus the exact
singleton sum, so it persists when pricing error is zero.}
\label{tab:3b_decomp}
\small
\begin{tabular}{@{}lrrrrr@{}}
\toprule
Prices & Price sum & Exact sum & Exact combination & Net pricing & Composition \\
\midrule
Midpoint & $-1.30611$ & $-1.30306$ & $-0.04571$ & $+0.00305$ & $+1.25735$ \\
Current & $-1.97173$ & $-1.28666$ & $-0.01139$ & $+0.68507$ & $+1.27527$ \\
Exact singleton & $-1.30426$ & $-1.30426$ & $-0.03671$ & $0.00000$ & $+1.26755$ \\
\bottomrule
\end{tabular}
\end{table}

The 132-change midpoint state's isolated ARC errors have absolute sum
$0.03014825$, MAE $0.00022840$ and maximum $0.00507330$; their signed sum is
$+0.00304859$. The composition residual is $41.71$ times the absolute sum.
This extends the 1B comparison in Appendix~\ref{app:pricing_l1} to a
healthy activation-ordered base.

For the midpoint state, the full TEST prediction error in
Equation~\eqref{eq:decomp}, including the data-shift term, separates as follows:
\begin{center}
\small
\begin{tabular}{@{}lrrrr@{}}
\toprule
Functional & FIT pricing & FIT composition & TEST $-$ FIT & Total error \\
\midrule
NLL & $-0.001713$ & $+0.019848$ & $+0.045734$ & $+0.063869$ \\
ARC-KL & $+0.003049$ & $+1.257347$ & $-0.005938$ & $+1.254457$ \\
\bottomrule
\end{tabular}
\end{center}
Composition dominates the ARC error. For TEST NLL, the data-shift term is
larger than the FIT composition term.
The two endpoint follow-ups use no new producer or backward passes and take
about 20 minutes on one RTX 5090, including loading and statistics.

\subsection{Support-capped selection on new 3B NLL data}
\label{app:3b_support_nll}

This comparison uses the activation-ordered Llama-3.2-3B bank above, with
static groups, INT3 codes and group size 128. All 196 projections start from
GPTQ($1\%$); each can retain that encoding or take the complete codes, scale
and zero tuple from RTN or GPTQ($0.1\%$). Non-projection weights and bitrate
remain fixed. FIT uses the original 28 NLL blocks and 592 ARC-Easy items.

We solve the exact-singleton additive objective with at most two, four or
eight projection changes: minimize summed ARC-KL differences subject to
summed NLL difference at most zero. Cap two is the primary control; caps
four and eight are fixed sensitivity arms. We also retain the original
134-change additive state, the base, and the two whole-model alternatives.
The cap-two beam has width four and a FIT slate capacity of four. Its
deterministic base-state, projection-index and target-choice ordering uses
588 charged queries: 392 archived singleton answers, each charged once,
followed by 196 newly measured pairs. The singleton queries complete the
base state's expansion. All new pairs extend the singleton $(11,2)$ with
both alternatives at projection indices $0$--$98$, excluding index 11.
The budget ends partway through this singleton's neighborhood, before
another singleton is expanded. Three fixed additive controls are measured
after search and do not enter its decisions. A baseline replay is also
measured separately. The additive selector uses 392 singleton prices;
the beam spends an additional 196 combination queries.

The beam primary minimizes FIT ARC-KL among its final slate states with
FIT NLL no larger than the base. Ties use NLL, then lexicographic state order;
an empty eligible set returns the base. The resulting primary changes
$(11,2)$ and $(16,1)$, while the additive cap-two state changes $(16,1)$ and
$(41,1)$. Here indices are zero-based projection indices, choice 1 is RTN,
and choice 2 is GPTQ($0.1\%$). Both states make exactly two changes.

\paragraph{Article-level TEST and pairing.}
All eight states are fixed from FIT before the new NLL readout, with no
intermediate VALIDATION selection. The dataset is drawn from the
WikiText-103 raw training split. We replay the original 128 calibration
and 28 FIT blocks to identify the 81 articles they touch, and exclude these
along with 114 articles in the later recorded 3B NLL evaluation.
Exclusion uses article identifiers, normalized titles and content matches.
The new set contains 256 complete 2,048-token blocks from 123 articles;
blocks never cross article boundaries. This separation is relative to those
recorded 3B article sets. Each state predicts 524,032 tokens.

For a comparison of states $s$ and $r$, let $L_b^s$ be the summed token loss
on block $b$, and $n_b=2{,}047$. We estimate
\[
 \widehat\delta_{s-r}
 =\frac{\sum_a\sum_{b\in a}(L_b^s-L_b^r)}
        {\sum_a\sum_{b\in a}n_b}.
\]
We resample the 123 articles with replacement, carrying every block in each
selected article and using identical article multiplicities across states.
The 95\% percentile intervals use 10,000 draws with seed 20260920.
The primary comparison is beam minus cap-two additive. This stage measures
NLL; ARC measurements remain on the original FIT units.

\begin{table}[!htb]
\centering
\caption{Support control on the new 3B NLL sample. Differences from the
GPTQ($1\%$) base, whose mean NLL is $2.28535013$; smaller is better.
Intervals pair states within article clusters. The last row is the
pre-specified comparison between the two support-two states.}
\label{tab:3b_support_nll}
\small
\setlength{\tabcolsep}{5pt}
\begin{tabular}{@{}lrrl@{}}
\toprule
State & Changes & $\Delta$ NLL & Paired $95\%$ interval \\
\midrule
Additive, cap 2 & 2 & $-0.00108$ & $[-0.00316,+0.00074]$ \\
Additive, cap 4 & 4 & $+0.00183$ & $[-0.00064,+0.00393]$ \\
Additive, cap 8 & 8 & $-0.00240$ & $[-0.00725,+0.00099]$ \\
Additive, unrestricted & 134 & $+0.06643$ & $[+0.04696,+0.09352]$ \\
Whole RTN & 196 & $+0.24065$ & $[+0.20646,+0.26444]$ \\
Whole GPTQ($0.1\%$) & 196 & $-0.00045$ & $[-0.03952,+0.03074]$ \\
Endpoint beam, cap 2 & 2 & $-0.00297$ & $[-0.00604,-0.00041]$ \\
\midrule
\multicolumn{2}{@{}l}{Beam $-$ cap-two additive} & $-0.00189$ & $[-0.00417,+0.00006]$ \\
\bottomrule
\end{tabular}
\end{table}

The original 134-change state's degradation recurs on this article set.
The capped additive states have small mean differences from the base,
with intervals spanning zero. The beam improves over the base, but its
paired advantage over cap-two addition remains unresolved.

\paragraph{Composition on the same FIT units.}
Table~\ref{tab:3b_support_composition} separates the singleton-sum prediction
from the executed combination. All differences use the same historical FIT
base, $(2.28081218,0.36947774)$ in NLL/ARC-KL order. For the three additive
caps, the ARC residual increases from $0.01415$ to $0.16662$ as more changes
are composed. At support two, addition gives lower FIT ARC-KL than the beam,
whereas the beam gives lower FIT NLL.

\begin{table}[!htb]
\centering
\caption{FIT composition effects for the support controls. Prediction is
the sum of exact singleton differences; residual is actual minus prediction.
These FIT differences are separate from the new TEST NLL differences.}
\label{tab:3b_support_composition}
\footnotesize
\setlength{\tabcolsep}{4pt}
\begin{tabular}{@{}lrrrrrr@{}}
\toprule
& \multicolumn{3}{c}{NLL} & \multicolumn{3}{c}{ARC option-KL} \\
\cmidrule(lr){2-4}\cmidrule(lr){5-7}
State & Predicted & Actual & Residual & Predicted & Actual & Residual \\
\midrule
Additive, cap 2 & $-0.00310$ & $-0.00338$ & $-0.00029$ & $-0.11010$ & $-0.09595$ & $+0.01415$ \\
Additive, cap 4 & $-0.00198$ & $-0.00170$ & $+0.00028$ & $-0.18124$ & $-0.12152$ & $+0.05972$ \\
Additive, cap 8 & $-0.00018$ & $-0.00912$ & $-0.00895$ & $-0.30585$ & $-0.13923$ & $+0.16662$ \\
Endpoint beam & $-0.02376$ & $-0.02144$ & $+0.00232$ & $-0.09353$ & $-0.09064$ & $+0.00289$ \\
\bottomrule
\end{tabular}
\end{table}

The runs use FP32 and SDPA, with TF32 disabled, on an RTX 6000D, without
training or regenerating the quantized alternatives. FIT takes 9,368.85 s,
including search and the fixed controls; the eight-state NLL measurement
takes 842.49 s, excluding loading and identity checks.

\subsection{Llama-3.2-1B with activation ordering}
The same producer on Llama-3.2-1B gives a bank of 224 replacements of the
act-order GPTQ($1\%$) codes by RTN or the act-order $0.1\%$ codes.  Exact
endpoints on the FIT units: $184/224$ raise NLL, $153/224$ raise ARC
option-KL, and the two directions conflict on $81/224$; midpoint signs
$223/224$ (NLL) and $222/224$ (ARC), current-state $50/224$ and $179/224$.

With the rules, splits and $352$-forward budget of the no-act-order bank,
the one-shot selectors change 81 (current), 37 (midpoint)
and 35 (exact singleton) projections with predicted FIT deltas
$(-0.460,-0.359)$, $(-0.00005,-0.093)$ and $(-0.0002,-0.088)$; their
VALIDATION deltas against the base are $(+0.424,+0.016)$, $(+0.021,-0.001)$
and $(+0.013,+0.014)$, TEST $(+0.339,+0.036)$, $(+0.033,-0.0005)$ and
$(+0.028,+0.016)$, the midpoint state's TEST NLL increase resolved
($[+0.006,+0.062]$) and its ARC change not ($[-0.016,+0.014]$).

The width-4
beam exhausted its budget after one expansion with four two-change FIT Pareto
states; the frozen validation rule selected $\{0,50\}\to\mathrm{RTN}$ (layer-0
q-projection, layer-7 k-projection).  Against the base: VALIDATION $-0.0004$
$[-0.0032,+0.0014]$ / $-0.0066$ $[-0.0115,-0.0017]$; TEST $-0.0014$
$[-0.0043,+0.0007]$ / $-0.0032$ $[-0.0081,+0.0013]$.  Against the one-shot
states on TEST: midpoint $-0.0345$ $[-0.0624,-0.0082]$ / $-0.0027$
$[-0.0167,+0.0115]$; exact singleton $-0.0295$ $[-0.0541,-0.0081]$ /
$-0.0193$ $[-0.0360,-0.0024]$; current $-0.341$ $[-0.429,-0.186]$ / $-0.0396$
$[-0.0586,-0.0203]$.  As on the no-act-order bank, the control was frozen and
run after the shared one-shot VALIDATION/TEST had been read.

\subsection{Qwen3-1.7B common-center abstentions}
The CommonsenseQA bank has 25 primitives covering $11.0$M physically changed,
proxy-improving groups. Of four stable FIT candidates, three pass
response-based outer confirmation. Own-midpoint rescoring reverses the
target response on all three, so FINAL remains unread.
Table~\ref{tab:driftfull} gives the full same-unit drift audit.

\begin{table}[htb]
\centering
\caption{All nine same-unit paired drift cells on the 1.7B bank
(Section~\ref{sec:drift}). Drift is proposal-center minus own-midpoint response.
Unit-level: support-1 drift negative on $32/32$
blocks and $100/100$ target microbatches; support-2 on $32/32$ and $98/100$;
support-4 positive on $32/32$ D0 and $29/32$ D1 blocks, negative on $100/100$
target microbatches.}
\label{tab:driftfull}
\vspace{0.4em}
\footnotesize
\setlength{\tabcolsep}{4pt}
\begin{tabular}{@{}llrrrl@{}}
\toprule
Candidate & Functional & Proposal center & Own center & Drift & Simultaneous CI \\
\midrule
support-1 & D0 NLL & $-0.00209$ & $+0.00249$ & $-0.00457$ & $[-0.00556,-0.00359]$ \\
support-1 & D1 NLL & $-0.00179$ & $+0.00311$ & $-0.00491$ & $[-0.00616,-0.00365]$ \\
support-1 & target KL & $-0.00219$ & $+0.00376$ & $-0.00595$ & $[-0.00642,-0.00548]$ \\
support-2 & D0 NLL & $-0.00186$ & $+0.00297$ & $-0.00482$ & $[-0.00603,-0.00362]$ \\
support-2 & D1 NLL & $-0.00134$ & $+0.00376$ & $-0.00509$ & $[-0.00658,-0.00360]$ \\
support-2 & target KL & $-0.00262$ & $+0.00262$ & $-0.00524$ & $[-0.00586,-0.00462]$ \\
support-4 & D0 NLL & $-0.05442$ & $-0.35554$ & $+0.30112$ & $[+0.22832,+0.37392]$ \\
support-4 & D1 NLL & $-0.04372$ & $-0.26543$ & $+0.22171$ & $[+0.12366,+0.31977]$ \\
support-4 & target KL & $-0.00355$ & $+0.06962$ & $-0.07317$ & $[-0.08033,-0.06600]$ \\
\bottomrule
\end{tabular}
\end{table}

The original WinoGrande protocol also returns no certified state. Candidates
$\{17,23\}$ and $\{17,22,23\}$ have target FIT means $-0.032$ and $-0.051$,
and outer full-center means $+0.025$ and $+0.014$. The fit-to-outer differences
are $+0.057$ and $+0.065$, with Welch intervals $[-0.011,+0.125]$ and
$[-0.003,+0.132]$. On the same outer units, medoid-center means are $+0.139$
and $+0.108$. Full-to-medoid center differences are $+0.114$ and $+0.094$,
with paired intervals $[+0.092,+0.137]$ and $[+0.073,+0.115]$. FINAL remains unread.

\section{Response Structure Across Domains}
\label{app:structure}

A $12\times52$ response matrix over six domains under NLL and full-vocabulary
teacher-KL is algebraically full rank with all twelve singular modes above a
block-bootstrap noise floor; participation rank $2.71$ with $75.9\%$ of energy
in two modes, $1.79$ (NLL) and $1.69$ (KL) within families; within-domain
NLL--KL correlation $0.15$--$0.63$.  Held-out completion under a frozen four-fold
split: nonnegative matched-family cone $78.5\%$ sign accuracy, Spearman $0.744$;
source identity $66.0\%$; ten-anchor unconstrained fit $72.3\%$ against $76.9\%$
for the strongest single anchor.  On eight moves frozen by generic quadrant
before any ARC-Easy output was read, KL-family signatures called $7/8$
gold-margin directions, the NLL family $3/8$.

\subsection{Teacher-KL is conditional on the teacher}
\label{app:teacher}

Stratifying option-KL effects by whether the teacher answers correctly inverts
the sign on three of the measured ARC-Easy moves between teacher-correct and
teacher-wrong strata.  Improving KL to a teacher transports its decisions,
including the wrong ones.

\section{Compute}
\label{app:compute}

On one RTX PRO 6000, Qwen3-0.6B's 23-primitive $W_C$ bank uses 28 D2 blocks
(shared model load $8.1$ s excluded). Setup costs $0.41$ s; shared-center
forward/backward $10.13$ s ($0.362$ s/block, 23 coordinates); baseline plus
23 primitive forwards $25.4$ s; sparse patches $93.9$ s; transfer $2.3$ s;
23 apply/revert operations $0.09$ s; CPU search over $880{,}969$ subsets
$51.4$ s. Search totals are $61.9$ s (response), $173.2$ s (sparse exact,
$2.80\times$), and $279.5$ s (full-artifact exact); break-even is $9.6$
primitives. Confirming B, E14 and the pair costs $3.19$ s. Peak allocated
memory is $21.3$ GB for responses, $10.3$ GB for sparse endpoints and
$4.9$ GB for a single primitive.
The 512-state atlas uses 84 functional-center sweeps; peak memory is
$26.7$ GB (D0), $13.0$ GB (ARC), $18.5$ GB (HellaSwag). Own-midpoint to
endpoint cost ratios are $3.63$ (4B median), $1.57$ (1B) and $1.50$ (3B
act-order); absolute timings accompany each bank in
Appendices~\ref{app:bank} and~\ref{app:protocols}.

\end{document}